\documentclass[3p,preprint]{elsarticle}
\usepackage[utf8]{inputenc}
\usepackage{amsthm,amssymb,amsmath}
\usepackage{hyperref}
\usepackage[capitalize]{cleveref}
\usepackage{todonotes}
\usepackage{xcolor}
\usepackage{thmtools}
\usepackage{color}

\def\dd{\mathinner{.\,.}}
\newcommand{\cO}{\mathcal{O}}
\newcommand{\RNK}{\textsf{RNK}}
\newcommand{\APP}{\textsf{APP}}
\newcommand{\PRE}{\textsf{PRE}}
\newcommand{\FTR}{\textsf{FTR}}

\newcommand{\pred}{\textsf{pred}}
\newcommand{\rank}{\textsf{rank}}
\newcommand{\select}{\textsf{select}}
\newcommand{\ISAW}{\textsf{ISAW}}
\DeclareMathOperator*{\argmin}{arg\,min}
\newcommand{\OP}{\textsf{SP}}

\newcommand{\CP}{\textsf{EP}}
\newcommand{\defproblem}[3]{
  \vspace{2mm}
\noindent\fbox{
  \begin{minipage}{0.96\textwidth}
  \textsc{#1}\\
  {\bf{Input:}} #2  \\
  {\bf{Output:}} #3
  \end{minipage}
  }
  \vspace{2mm}
}

\theoremstyle{plain}
\newtheorem{theorem}{Theorem}
\newtheorem{lemma}{Lemma}

\newtheorem{proposition}{Proposition}
\newtheorem{fact}[theorem]{Fact}
\newtheorem*{claim*}{Claim}
\theoremstyle{definition}
\newtheorem{definition}{Definition}
\newtheorem{example}{Example}

\begin{document}
\begin{frontmatter}
\title{Internal Shortest Absent Word Queries in Constant Time and Linear Space\tnoteref{t1}}
\author[1]{Golnaz Badkobeh}\ead{g.badkobeh@gold.ac.uk}
\author[2]{Panagiotis Charalampopoulos\footnote{Supported by the Israel Science Foundation grant 592/17.}}\ead{panagiotis.charalampopoulos@post.idc.ac.il}
\author[3]{Dmitry Kosolobov\footnote{Supported by the Ministry of Science and Higher Education of the Russian Federation (Ural Mathematical Center project No. 075-02-2021-1387)}}\ead{dkosolobov@mail.ru}
\author[4,5]{Solon P.~Pissis}\ead{solon.pissis@cwi.nl}
\address[1]{Department of Computing, Goldsmiths University of London, UK}
\address[2]{Efi Arazi School of Computer Science, The Interdisciplinary Center Herzliya, Israel}
\tnotetext[t1]{The present paper is an extended and improved version of an earlier text that appeared in the 32nd Annual Symposium on Combinatorial Pattern Matching, CPM 2021~\cite{self:CPM}.}
\address[3]{Ural Federal University, Ekaterinburg, Russia}
\address[4]{CWI, Amsterdam, The Netherlands}
\address[5]{Vrije Universiteit, Amsterdam, The Netherlands}
\begin{abstract}
Given a string $T$ of length $n$ over an alphabet $\Sigma\subset \{1,2,\ldots,n^{\cO(1)}\}$ of size $\sigma$, we are to preprocess $T$ so that given a range $[i,j]$, we can return a representation of a shortest string over $\Sigma$ that is absent in the fragment $T[i]\cdots T[j]$ of~$T$. We present an $\cO(n)$-space data structure that answers such queries in constant time and can be constructed in $\cO(n\log_\sigma n)$ time.
\end{abstract}
\end{frontmatter}

\section{Introduction}\label{sec:intro}

Range queries are a classic data structure topic~\cite{Yao1982,DBLP:journals/siamcomp/BerkmanV93,DBLP:conf/latin/BenderF00}. In 1d, a range query $q=f(A,i,j)$ on an array of $n$ elements over some set $U$, denoted by $A[1\dd n]$, takes two indices $1 \leq i \leq j \leq n$, a function $f$ defined over arrays of elements of $U$, and outputs $f(A[i\dd j]) = f(A[i],\ldots,A[j])$.
Range query data structures in 1d
can thus be viewed as data structures answering queries on a string in the internal setting, where $U$ is the considered alphabet.

Internal queries on a string have received much attention in recent years. In the internal setting, we are asked to preprocess a string $T$ of length $n$ over an alphabet $\Sigma$ of size $\sigma$, so that queries about substrings of $T$ can be answered efficiently. Note that an arbitrary substring of $T$ can be encoded in $\cO(1)$ words of space by the indices $i,j$ of its occurrence as a fragment $T[i]\cdots T[j]=T[i \dd j]$ of $T$. Data structures for answering internal queries are interesting in their own right, but also have numerous applications in the design of algorithms and (more sophisticated) data structures. Because of these numerous applications, we usually place particular emphasis on the construction time---other than on the tradeoff between space and query time, which is the main focus in the classic data structure literature.

In data structures on strings it is typically assumed that the input alphabet is integer and polynomially bounded, i.e., it is a subset of $\{1,2,\ldots,n^{\cO(1)}\}$ where $n$ is the length of the input string $T$.
One of the most widely-used internal queries is that of asking for the \emph{longest common prefix} of two suffixes $T[i \dd n]$ and $T[j \dd n]$ of~$T$.
The classic data structure for this problem~\cite{DBLP:journals/jcss/LandauV88} consists of the suffix tree of $T$~\cite{DBLP:conf/focs/Farach97} and a lowest common ancestor data structure~\cite{DBLP:journals/siamcomp/HarelT84} over the suffix tree. It occupies $\cO(n)$ space, it can be constructed in $\cO(n)$ time, and it answers queries in $\cO(1)$ time.
In the word RAM model of computation with word size $\Theta(\log n)$ bits the construction time is not necessarily optimal when the input alphabet is $\{1,2,\ldots,\sigma\}$ and the string is packed into $\cO(n/\log_\sigma n)$ machine words. A sequence of works~\cite{DBLP:conf/mfcs/TanimuraNBIT17,DBLP:conf/cpm/MunroNN20,DBLP:conf/soda/BirenzwigeGP20} has culminated in the recent optimal data structure of Kempa and Kociumaka~\cite{DBLP:conf/stoc/KempaK19}: it occupies $\cO(n/\log_\sigma n)$ space, it can be constructed in $\cO(n/\log_\sigma n)$ time, and it answers queries in $\cO(1)$ time.

Another fundamental problem in this setting is the \emph{internal pattern matching} (IPM) problem.
It consists in preprocessing $T$ so that we can efficiently compute the occurrences of a substring $U$ of $T$ in another substring $V$ of $T$.
For the decision version of the IPM problem,
Keller et al.~\cite{DBLP:journals/tcs/KellerKFL14} presented a data structure of nearly-linear size supporting sublogarithmic-time queries.
Kociumaka et al.~\cite{DBLP:conf/soda/KociumakaRRW15}
presented a data structure of linear size supporting constant-time queries when the ratio between the lengths of $V$ and $U$ is bounded by a constant.
The $\cO(n)$-time construction algorithm of the latter data structure was derandomized in~\cite{tomeksthesis}.
In fact, Kociumaka et al.~\cite{DBLP:conf/soda/KociumakaRRW15}, using their efficient IPM queries as a subroutine, managed to show efficient solutions for other internal problems, such as for computing the periods of a substring (\emph{period queries}, introduced in~\cite{DBLP:conf/spire/KociumakaRRW12}), and for checking whether two substrings are rotations of one another
(\emph{cyclic equivalence queries}).
Other problems that have been studied in the internal setting include string alignment~\cite{DBLP:journals/mics/Tiskin08,DBLP:journals/corr/abs-2103-03294}, approximate pattern matching~\cite{DBLP:conf/focs/Charalampopoulos20}, dictionary matching~\cite{DBLP:conf/isaac/Charalampopoulos19,DBLP:conf/cpm/Charalampopoulos20}, longest common substring~\cite{DBLP:journals/algorithmica/AmirCPR20}, counting palindromes~\cite{DBLP:conf/spire/RubinchikS17}, range longest common prefix~\cite{DBLP:journals/jcss/AmirALLLP14,DBLP:conf/cocoon/Abedin0HNSST18,matsuda2020compressed,DBLP:journals/fuin/GangulyPST18}, the computation of the lexicographically minimal or maximal suffix, and minimal rotation~\cite{DBLP:journals/tcs/BabenkoGKKS16,DBLP:conf/cpm/Kociumaka16}, as well as of the lexicographically $k$th suffix~\cite{DBLP:conf/soda/BabenkoGKS15}.
We refer the interested reader to the Ph.D dissertation of Kociumaka~\cite{tomeksthesis}, for a nice exposition.

In this work, we extend this line of research by investigating the following basic internal query, which, to the best of our knowledge, has not been studied previously. Given a string $T$ of length $n$ over an alphabet $\Sigma\subset \{1,2,\ldots,n^{\cO(1)}\}$, preprocess $T$ so that given a range $[i,j]$, we can return a shortest string over $\Sigma$ that does not occur in $T[i\dd j]$.
The latter shortest string is also known as a shortest absent word in the literature.
We work on the standard unit-cost word RAM model with machine word-size $w= \Theta(\log n)$ bits. We measure the space used by our algorithms and data structures in machine words, unless stated otherwise. We assume that we have random access to $T$ and so our algorithms return a constant-space representation of a shortest string (a witness) consisting of a substring of $T$ and a letter.
A na\"ive solution for this problem precomputes a table of size $\cO(n^2)$ that stores the answer for every possible query $[i,j]$. Our main result is the following theorem.

\begin{restatable}{theorem}{mainthm}\label{the:main}
Given a string $T$ of length $n$ over an alphabet $\Sigma\subset \{1,2,\ldots,n^{\cO(1)}\}$ of size~$\sigma$, we can construct in $\cO(n \log_{\sigma} n)$ time a data structure of size $\cO(n)$ that, for any given query $[a,b]$, can compute in $\cO(1)$ time a shortest string over $\Sigma$ that does not occur in $T[a\dd b]$.
\end{restatable}


In an earlier conference version of the present paper~\cite{self:CPM}, we have obtained a weaker result: a data structure of size $\cO((n/k)\cdot \log\log_\sigma n)$ that can answer queries in $\cO(\log\log_\sigma k)$ time, where $k$ is a user-defined parameter from $[1, \log\log_\sigma n]$.
The improved data structure presented in this manuscript combines ideas from the conference version and the utilization of succinct fusion trees introduced by Grossi et al.~\cite{grossi_et_al:LIPIcs:2009:1847}.

In the related \emph{range shortest unique substring} problem, defined by Abedin et al.~\cite{DBLP:journals/algorithms/Abedin0PT20}, the task is to construct a data structure over $T$ to be able to answer the following type of online queries efficiently. Given a range $[i,j]$, return a shortest string with exactly one occurrence (starting position) in $[i,j]$. Abedin et al.~presented a data structure of size $\cO(n\log n)$ supporting $\cO(\log_w n)$-time queries, where $w=\Theta(\log n)$ is the word size. Additionally, Abedin et al.~\cite{DBLP:journals/algorithms/Abedin0PT20} presented a data structure of size $\cO(n)$ supporting $\cO(\sqrt{n}\log^\epsilon n)$-time queries, where $\epsilon$ is an arbitrarily small positive constant.

\paragraph*{Our Techniques}

For clarity of exposition, in this overview, we skip the time-efficient construction algorithms of our data structures and only describe how to compute the \emph{length} of a shortest absent word (without a witness) in $T[a\dd b]$; note that this length is at most $\log_\sigma n$. Let us also recall that the length of a shortest absent word of $T$ can be computed in $\cO(n)$ time using the suffix tree of~$T$~\cite{DBLP:conf/focs/Farach97}. It suffices to traverse the suffix tree of $T$ recording the shortest string-depth $\ell$, where an implicit or explicit node has less than $\sigma$ outgoing edges.

\emph{First approach:} We precompute, for each position $i$ and for each length $j \in [1, \log_\sigma n]$, the starting position of the shortest suffix of $T[1 \dd i]$ that contains an occurrence of each of the $\sigma^j$ distinct words of length $j$. Then, a query for the length of a shortest absent word of $T[a \dd b]$ reduces to computing the predecessor of $a$ among the starting positions we have precomputed for position $b$. By maintaining these $\cO(\log_\sigma n)$ starting positions in a fusion tree~\cite{DBLP:journals/jcss/FredmanW93}, we obtain a data structure of size $\cO(n \log_\sigma n)$ supporting queries in $\cO(\log_w \log n) = \cO(1)$ time.

\emph{Second approach:} We precompute, for each length $j \in [1, \log_\sigma n]$, all minimal fragments of $T$ that contain an occurrence of each of the distinct $\sigma^j$ words of length $j$.
As these fragments are inclusion-free, we can encode them using two $n$-bit arrays storing their starting and ending positions in $T$, respectively. We thus require $\cO(n)$ words of space in total over all~$j$s. Observe that $T[a \dd b]$ does not have an absent word of length $j$ if and only if it contains a minimal fragment for length $j$; we can check this condition in $\cO(1)$ time after augmenting the computed bit arrays with succinct rank and select data structures~\cite{DBLP:conf/focs/Jacobson89}. Finally, due to monotonicity (if $T[a\dd b]$ contains all strings of length $j+1$ then $T$ contains all strings of length $j$), we can binary search for the answer in $\cO(\log \log_\sigma n)$ time.

\emph{Third approach:}
We optimize the first approach by utilizing succinct fusion trees to store the sets of size $\cO(\log_\sigma n)$ associated with positions of $T$, thus reducing the space on top of the sets to $\cO(n\log_\sigma\log n)$. Instead of storing the $\cO(\log_\sigma n)$-size sets explicitly, we compute their elements on demand using $\cO(\log_\sigma n)$ select data structures, each occupying $\cO(n)$ bits. This leads to an $\cO(n\log_\sigma \log n)$-space solution.
In order to optimize it further, we rely on the following combinatorial observation: if the length of a shortest absent word of a string $X$ over $\Sigma$ is $\lambda$, we need to append $\Omega(\sigma^{d-1} \cdot \lambda)$ letters to $X$ in order to obtain a string with a shortest absent word of length $\lambda+d$.
(For intuition, think of $|X|$ as a constant; then, we essentially need to append the de Bruijn sequence of order $d$ over $\Sigma$ to $X$ in order to achieve the desired result.)
This observation allows us to lower the memory consumption by truncating all succinct fusion trees at positions that are not multiples of $\log\log n$, by building them only for their first $\cO(\log n / \log\log n)$ entries. The total space thus reduces to $\cO(n)$ words. A query for the length of a shortest absent word of $T[a \dd b]$ is performed by first checking whether the answer is at most $\log n /\log\log n$, which is done using the (truncated) fusion tree stored at $b$, and, if not, a query on $T[a \dd b']$ is performed, where $b'$ is the closest multiple of $\log\log n$ after $b$. It can be shown using the combinatorial observation that the answer for $T[a \dd b]$ is within an $\cO(1)$-length range of the answer for $T[a \dd b']$, and it is computed by the data structure from the second approach.

%

\paragraph*{Other Related Work}

Let us recall that a string $S$ that does not occur in $T$ is called \emph{absent} from $T$, and if all its proper substrings appear in $T$ it is called a \emph{minimal absent word} of $T$. It should be clear that every shortest absent word is also a minimal absent word. Minimal absent words (MAWs) are used in many applications~\cite{Silva02042015,10.1093/bioinformatics/btaa686,DBLP:journals/tcs/FiciMRS06,minabpro,Chairungsee2012109,DBLP:conf/isita/OtaM10,DCA} and their theory is well developed~\cite{DBLP:journals/tcs/MignosiRS02,DBLP:conf/spire/FiciG19,FiReRi19}, also from an algorithmic and data structure point of view~\cite{DBLP:conf/sofsem/MienoKAFNIBT20,DBLP:journals/iandc/CrochemoreHKMPR20,MAW,DBLP:conf/spire/Charalampopoulos18,IC18,dcc19,DBLP:conf/mfcs/FujishigeTIBT16,DBLP:conf/ppam/BartonHMP15,Crochemore98automataand}. For example, it is well known that, given two strings $X$ and $Y$, one has $X=Y$ if and only if $X$ and $Y$ have the same set of MAWs~\cite{DBLP:journals/tcs/MignosiRS02}.

\paragraph*{Paper Organization}
Section~\ref{sec:prel} provides some preliminaries.
The first approach is detailed in Section~\ref{sec:constant} and the second one in Section~\ref{sec:loglog}.
Section~\ref{sec:comb} provides the combinatorial foundations for the third approach, which is detailed in Section~\ref{sec:main}. Sections~\ref{sec:constant}--\ref{sec:comb} have essentially already appeared in the conference version~\cite{self:CPM} of our paper; the main difference and novelty lie in Section~\ref{sec:main}. We conclude with  open problems in Section~\ref{sec:fin}.


\section{Preliminaries}\label{sec:prel}

An \emph{alphabet} $\Sigma$ is a finite nonempty set whose elements are called \emph{letters}.
A {\em string} (or \emph{word}) $S = S[1\dd n]$ is a sequence of \emph{length} $|S|=n$ over~$\Sigma$.
The {\em empty} string $\varepsilon$ is the string of length $0$.
The \emph{concatenation} of two strings $S$ and $T$ is the
string composed of the letters of $S$ followed by the letters of $T$; it is denoted by $S\cdot T$ or simply by $ST$.
The set of all strings (including $\varepsilon$) over $\Sigma$ is denoted by $\Sigma^{*}$.
The set of all strings of length $k>0$ over $\Sigma$ is denoted by $\Sigma^{k}$. For $1 \leq i \le j \leq n$, $S[i]$ denotes the $i$th letter of $S$, and the fragment $S[i\dd j]$ denotes an \emph{occurrence} of the underlying \emph{substring} $P=S[i]\cdots S[j]$. We say that $P$ {\em occurs} at (starting) \emph{position} $i$ in $S$. A string $P$ is called \emph{absent} from $S$ if it does not occur in $S$.
A substring $S[i \dd j]$ is a \emph{suffix} of $S$ if $j=n$ and it is a \emph{prefix} of $S$ if $i=1$.

The following proposition is straightforward (as explained in Section~\ref{sec:intro}).

\begin{proposition}
Let $T$ be a string of length $n$ over an alphabet $\Sigma\subset \{1,2,\ldots,n^{\cO(1)}\}$.
A shortest absent word of $T$ can be computed in $\cO(n)$ time.
\label{th:saw}
\end{proposition}

Given an array $A$ of $n$ items taken from a totally ordered set, the \emph{range minimum query} $\textsf{RMQ}_A(\ell,r) =\argmin A[k]$ (with $1 \leq \ell \leq k \leq r \leq n$) returns the position of the minimal element in $A[\ell \dd r]$. The following result is known.

\begin{theorem}[\cite{DBLP:conf/latin/BenderF00,DBLP:journals/siamcomp/FischerH11}]
Let $A$ be an array of $n$ integers.
A data structure of size $2n+o(n)$ bits that supports $\textsf{RMQ}$s on $A$ in $\cO(1)$ time without the need to store and access $A$ itself can be constructed in $\cO(n)$ time.%
\label{th:rmq}%
\end{theorem}

We make use of {\em rank and select} data structures constructed over bit vectors. For a bit vector $H$ we define $\textsf{rank}_q(i,H)=|\{k\in[1,i]:H[k]=q\}|$ and $\textsf{select}_q(i,H)=\min\{k\in[1,n]:\textsf{rank}_q(k,H)=i\}$, for $q\in\{0,1\}$. The following result is known.

\begin{theorem}[\cite{DBLP:conf/focs/Jacobson89,DBLP:books/daglib/0038982}]
Let $H$ be a bit vector of $n$ bits.
A data structure of $o(n)$ additional bits that supports $\textsf{rank}$ and $\textsf{select}$ queries on $H$ in $\cO(1)$ time can be constructed in $\cO(n)$ time.
\label{th:rs}
\end{theorem}

The \emph{static predecessor} problem consists in preprocessing a set $Y$ of integers, over an ordered universe~$U$, so that, for any integer $x \in U$ one can efficiently return the predecessor $\pred(x):=\max\{y \in Y : y \leq x\}$ of $x$ in $Y$.
The successor problem is defined analogously: upon a queried integer $x \in U$, the successor $\min\{y \in Y : y \geq x\}$ of $x$ in $Y$ is to be returned.
Willard and Fredman designed the \emph{fusion tree} data structure for this problem~\cite{DBLP:journals/jcss/FredmanW93}.
In the dynamic variant of the problem, updates to $Y$ are interleaved with predecessor and successor queries.
P{\u{a}}tra{\c{s}}cu and Thorup~\cite{DBLP:conf/focs/PatrascuT14} presented a dynamic version of fusion trees, which, in particular, yields an efficient construction of this data structure.

\begin{theorem}[\cite{DBLP:journals/jcss/FredmanW93,DBLP:conf/focs/PatrascuT14}]\label{thm:fusion}
Let $Y$ be a set of at most $n$ $w$-bit integers.
A data structure of size $\cO(n)$ can be constructed in $\cO(n\log_w n)$ time supporting insertions, deletions, and predecessor queries on $Y$ in $\cO(\log_w n)$ time.
\end{theorem}

We also use a succinct version of the (static) fusion tree that utilizes only $\cO(n\log w)$ \emph{bits} on top of a read-only array $Y$ of length $n$ (in contrast, the fusion tree from Theorem~\ref{thm:fusion} uses $\cO(nw)$ \emph{bits}). In this data structure there is no need to store the array $Y$ explicitly. Instead, $Y$ can be ``emulated'' by computing its elements on demand in $\cO(1)$ time.
Albeit it is not explicitly stated in~\cite{grossi_et_al:LIPIcs:2009:1847,DBLP:conf/compgeom/ChanLP11}, it follows from their construction that the succinct version can be constructed from a (usual) fusion tree in linear time.


\begin{theorem}[\cite{grossi_et_al:LIPIcs:2009:1847,DBLP:conf/compgeom/ChanLP11}]\label{thm:fusion-succinct}
Let $Y$ be a read-only array of at most $n$ $w$-bit integers and $n \le w^{\cO(1)}$.
A data structure of size $\cO(n\log w)$ bits can be constructed in $\cO(n\log_w n)$ time supporting predecessor queries on the elements of $Y$ in $\cO(\log_w n)$ time, provided that a table computable in $o(2^w)$ time and independent of the array has been precomputed.
\end{theorem}

Note that if we build multiple predecessor queries for sets of $w$-bit integers using the above theorem, they can all share a unique table computable in $o(2^w)$ time.

If $|U|=\cO(n)$, then, after an $\cO(n)$-time preprocessing, we can answer predecessor queries over the integer universe $U$ in $\cO(1)$ time as follows.
For each $y \in Y$, we set the $y$th bit of an initially all-zeros $|U|$-size bit vector. We then preprocess this bit vector as in~\cref{th:rs}.
Then, a predecessor query for any integer $x$ can be answered in $\cO(1)$ time due to the following readily verifiable formula:
$\pred(x)=\select_1(\rank_1(x))$.

The main problem considered in this paper is formally defined as follows.

\defproblem{\textsf{Internal Shortest Absent Word} (\ISAW)}{A string $T$ of length $n$ over an alphabet $\Sigma\subset \{1,2,\ldots,n^{\cO(1)}\}$ of size $\sigma>1$.}
 {{Given integers $a$ and $b$, with $1\leq a \leq b \leq n$, output a shortest string in $\Sigma^*$ with no occurrence in $T[a \dd b]$}.}

\noindent If $a=b$ then the answer is trivial. So, in what follows we assume that $a<b$. Let us also remark that the output (shortest absent word) can be represented in $\cO(1)$ space using: either a range $[i, j] \subseteq [1,n]$ and a letter $\alpha$ of $\Sigma$, such that the shortest string in $\Sigma^*$ with no occurrence in $T[a \dd b]$ is $T[i \dd j] \alpha$; or simply a range $[i, j] \subseteq [1,n]$ such that the shortest string in $\Sigma^*$ with no occurrence in $T[a \dd b]$ is $T[i \dd j]$.

\begin{example}
Given the string $T = \texttt{abaabaa\textcolor{red}{abbabbb}aaab}$ and the range $[a, b]= [8, 14]$ (shown in red), the only shortest absent word of $T[8 \dd 14]$ is $T[i\dd j]=T[7\dd 8]=\texttt{aa}$.
\end{example}

\section{\boldmath $\cO(n\log_\sigma n)$ Space and $\cO(1)$ Query Time}\label{sec:constant}

Let $T$ be a string of length $n$.
We define $S_T(j)$ as the function counting the cardinality of the set of length-$j$ substrings of $T$. This is known as the \emph{substring complexity} function~\cite{DBLP:journals/dm/Ferenczi99,DBLP:journals/algorithmica/RaskhodnikovaRRS13}. Note that $S_T(j)\leq n$, for all $j$. We have the following simple fact.

\begin{fact}\label{fct:saw}
The length $\ell$ of a shortest absent word of a string $T$ of length $n$ over an alphabet of size $\sigma$ is equal to the smallest $j$ for which $S_T(j)<\sigma^j$ and hence $\ell\in [1,\lfloor \log_\sigma n \rfloor]$.
\end{fact}

We denote the set of shortest absent words of $T$ by $\text{SAW}_T$.
Recall that, by Proposition~\ref{th:saw}, a shortest absent word of $T$ can be computed in $\cO(n)$ time.
We denote the length of the shortest absent words of $T$ by~$\ell$. By Fact~\ref{fct:saw}, $\ell\leq \lfloor\log_\sigma n \rfloor$. Since $\ell$ is an upper bound on the length of the answer for any \ISAW{} query on $T$, in what follows, we consider only lengths in $[1,\ell-1]$. Let one such length be denoted by $j$. By constructing and traversing the suffix tree of $T$, we can assign to each $T[i \dd i+j-1]$ its lexicographic rank in $\Sigma^j$. The time required for each length $j$ is $\cO(n)$, since the suffix tree of $T$ can be constructed within this time~\cite{DBLP:conf/focs/Farach97}.
Thus, the total time for all lengths $j\in[1,\ell-1]$ is $\cO(n \log_\sigma n)$ by Fact~\ref{fct:saw}.

We design the following warm-up solution to the \ISAW{} problem. For all $j\in[1,\ell-1]$ we store an array $\RNK_j$ of $n$ integers such that $\RNK_j[i]$ is equal to the lexicographic rank of $T[i \dd i+j-1]$ in $\Sigma^j$. Then, given a range $[a,b]$, in order to check if there is an absent word of length $j$ in $T[a \dd b]$ we only need to compute the number of distinct elements in $\RNK_j[a\dd b-j+1]$.
It is folklore that using a persistent segment tree, we can preprocess an array $A$ of $n$ integers in $\cO(n \log n)$ time so that upon a range query $[a,b]$ we can return the number of distinct elements in $A[a\dd b]$ in $\cO(\log n)$ time.
Thus, we could use this tool as a black box for every array $\RNK_j$ resulting, however, in $\Omega(\log n)$-time queries. We improve upon this solution as follows.

We employ a range minimum query (RMQ) data structure~\cite{DBLP:conf/latin/BenderF00} over a slight modification of $\RNK_j$.
For each $j$, we have an auxiliary procedure checking whether all strings from $\Sigma^j$ occur in $T[a \dd b]$ or not (i.e., it suffices to check whether any lexicographic rank is absent from the corresponding range). Similar to the previous solution, we rank the elements of $\Sigma^j$ by their lexicographic order.
We append $\RNK_j$ with all integers in $[1, \sigma^j]$. Let this array be $\APP_j$. By Fact~\ref{fct:saw}, we have that $|\APP_j|\leq 2n$. Then, we construct an array $\PRE_j$ of size $|\APP_j|$: $\PRE_j[i]$ stores the position of the rightmost occurrence of $\APP_j[i]$ in $\APP_j[1\dd i-1]$ (or $0$ if such an occurrence does not exist). This can be done in $\cO(n)$ time per $j$ by sorting the list of pairs $(T[i\dd i+j-1],i)$, for all $i$, using the suffix tree of $T$ to assign ranks for $T[i\dd i+j-1]$ and then radix sort to sort the list of pairs.

We now rely on the following fact.

\begin{fact}\label{fac:trick}
$S_{T[a\dd b]}(j)=\sigma^j$
if and only if $\min\{\PRE_j[i]:i\in[b-j+2,|\PRE_j|]\}\geq a$.
\end{fact}
\begin{proof}
If the smallest element in $\PRE_j[b-j+2 \dd |\PRE_j|]$, say $\PRE_j[k]$, is such that $\PRE_j[k]\geq a$, then all ranks of elements in $\Sigma^j$ occur in $\APP_j[a \dd b-j+1]$. This is because all elements (ranks) in $\Sigma^j$
occur at least once after $b-j+2$ (due to appending all integers in $[1, \sigma^j]$ to $\RNK_j$), and thus all must have a representative occurrence after $b-j+2$. Inspect Figure~\ref{fig:fact9} for an illustration. (The opposite direction is analogous.)
\end{proof}

\begin{figure}[h!]
    \centering
    \begin{center}
\includegraphics[width=9cm]{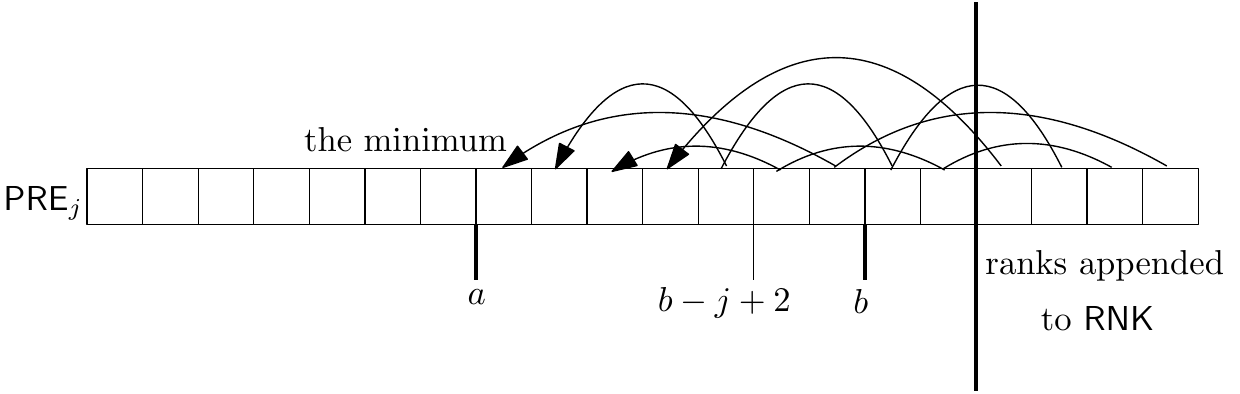}
\end{center}
    \caption{Illustration of the setting in \cref{fac:trick}.}
    \label{fig:fact9}
\end{figure}

The following two examples illustrate the construction of arrays $\RNK_j$, $\APP_j$, and $\PRE_j$ as well as \cref{fac:trick}.

\begin{example}[Construction] Let $T= \tt abaabaaabbabbbaaab$ and $\Sigma=\{\texttt{a},\texttt{b}\}$.
The set $\text{SAW}_T$ of shortest absent words of $T$ over $\Sigma$, each of length $\ell=4$, is $\{\tt aaaa, \tt abab, \tt baba, \tt bbbb\}$. Arrays $\RNK_j$, $\APP_j$, and $\PRE_j$, for all $j\in[1,\ell-1]$, are as depicted in Table~\ref{tbl:ex-con}.
\begin{table}[ht]
\caption{Arrays $\RNK_j$, $\APP_j$, and $\PRE_j$ in Example~\ref{ex:con}.}
\vspace{1mm}
\scalebox{0.89}{
\begin{tabular}{l|l*{25}{p{5pt}}l@{}}
     $i$&\tt1&\tt2&\tt3&\tt4&\tt5&\tt6&\tt7&\tt8&\tt9&\tt10&\tt11 &\tt12&\tt13&\tt14&\tt15&\tt16&\tt17&\tt18&\tt19&\tt20&\tt21&\tt22&\tt23&\tt24 \\
     $T$&\tt a&\tt b&\tt a&\tt a&\tt b&\tt a&\tt a&\tt a&\tt b&\tt b&\tt a & \tt b&\tt b&\tt b &\tt a&\tt a&\tt a&\tt b \\
     \hline
     $\RNK_1$&\tt 1&\tt 2&\tt 1&\tt 1&\tt 2&\tt 1&\tt 1&\tt 1&\tt 2&\tt 2&\tt 1&\tt 2&\tt 2&\tt 2&\tt 1&\tt 1&\tt 1&\tt 2\\
     $\APP_1$&\tt 1&\tt 2&\tt 1&\tt 1&\tt 2&\tt 1&\tt 1&\tt 1&\tt 2&\tt 2&\tt 1&\tt 2&\tt 2&\tt 2&\tt 1&\tt 1&\tt 1&\tt 2&\tt 1&\tt 2 \\
     $\PRE_1$&\tt0&\tt0&\tt 1&\tt 3&\tt 2&\tt4&\tt 6&\tt 7&\tt 5&\tt 9&\tt 8&\tt 10&\tt 12&\tt 13&\tt 11&\tt 15&\tt 16&\tt 14&\tt 17&\tt 18\\\hline
     $\RNK_2$&\tt 2&\tt 3&\tt 1&\tt 2&\tt 3&\tt 1&\tt 1&\tt 2&\tt 4&\tt 3&\tt 2&\tt 4&\tt 4&\tt 3&\tt 1&\tt 1&\tt 2\\
     $\APP_2$&\tt 2&\tt 3&\tt 1&\tt 2&\tt 3&\tt 1&\tt 1&\tt 2&\tt 4&\tt 3&\tt 2&\tt 4&\tt 4&\tt 3&\tt 1&\tt 1&\tt 2&\tt 1&\tt 2&\tt 3&\tt 4\\
     $\PRE_2$&\tt0&\tt 0&\tt0&\tt 1&\tt 2&\tt 3&\tt 6&\tt 4&\tt 0&\tt5&\tt 8&\tt 9&\tt 12&\tt 10&\tt 7&\tt 15&\tt 11&\tt 16&\tt 17&\tt 14&\tt 13\\
     \hline
     $\RNK_3$&\tt3&\tt 5&\tt2&\tt 3&\tt 5&\tt 1&\tt 2&\tt 4&\tt 7&\tt6&\tt 4&\tt 8&\tt 7&\tt 5&\tt 1&\tt 2\\
     $\APP_3$&\tt3&\tt 5&\tt2&\tt 3&\tt 5&\tt 1&\tt 2&\tt 4&\tt 7&\tt6&\tt 4&\tt 8&\tt 7&\tt 5&\tt 1&\tt 2&\tt 1&\tt 2&\tt 3&\tt 4&\tt 5&\tt6&\tt 7&\tt 8\\
     $\PRE_3$&\tt0&\tt 0&\tt0&\tt 1&\tt 2&\tt 0&\tt 3&\tt0&\tt 0&\tt0&\tt 8&\tt 0&\tt 9&\tt 5&\tt 6&\tt7&\tt 15&\tt16&\tt 4&\tt 11&\tt 14&\tt10&\tt 13&\tt 12\\
 \end{tabular}
 }\label{tbl:ex-con}
 \end{table}
 For instance, $\RNK_2[15]=\APP_2[15]=1$ denotes that the lexicographic rank of $\texttt{aa}$ in $\Sigma^2$ is~$1$; and $\PRE_2[15]=7$ denotes that the previous rightmost occurrence of $\texttt{aa}$ is at position~$7$.
\label{ex:con}
\end{example}

\begin{example}[Fact~\ref{fac:trick}]
Let $[a,b]=[7,11]$ and $j=2$ (see Example~\ref{ex:con}).
The smallest element in $\{\PRE_2[11],\ldots,\PRE_2[21]\}$ is  $\PRE_2[15]=7\geq a=7$, which corresponds to rank $\APP_2[15]=1$. Indeed all other ranks $2,3,4$ have at least one occurrence within $\APP_2[7\dd 11]=1,2,4,3,2$.
\end{example}

To apply Fact~\ref{fac:trick}, we construct, in $\cO(n)$ time, an $\cO(n)$-space, $\cO(1)$-query-time RMQ data structure over $\PRE_j$; see Theorem~\ref{th:rmq}. This results in $\cO(n\ell)=\cO(n\log_\sigma n)$ preprocessing time and space over all $j$.

For querying, let us observe that $\sigma^j-S_{T[a\dd b]}(j)$, for any $T,a,b$ and increasing $j$, is non-decreasing. We can thus apply binary search on $j$ to find the smallest length $j$ such that $S_{T[a\dd b]}(j)<\sigma^j$. This results in $\cO(\log\ell)=\cO(\log\log_\sigma n)$ query time.
We obtain the following proposition (retrieving a witness shortest absent word is detailed later).

\begin{proposition}
Given a string $T$ of length $n$ over an alphabet $\Sigma\subset \{1,2,\ldots,n^{\cO(1)}\}$ of size $\sigma$, we can construct a data structure of size $\cO(n \log_{\sigma} n)$ in $\cO(n \log_{\sigma} n)$ time, so that if query $[a,b]$ is given, we can compute a shortest string over $\Sigma$ that does not occur in $T[a\dd b]$ in $\cO(\log\log_\sigma n)$ time.
\end{proposition}

We further improve the query time via employing fusion trees as follows. We create a 2d array $\FTR[1\dd \ell-1][1\dd n]$ of integers, where \[\FTR[j][i]=\min\{\PRE_j[i-j+2],\ldots,\PRE_j[|\PRE_j|]\},\]
for all $j \in [1,\ell-1]$ and $i\in[1,n]$. Intuitively, $\FTR[j][i]$ is the rightmost index of $T$ such that $T[\FTR[j][i]\dd i]$ contains all strings of length $j$ over $\Sigma$ if such an index exists and $0$ otherwise.

Array $\FTR$ can be constructed in $\cO(n\ell)=\cO(n\log_\sigma n)$ time by scanning each array $\PRE_j$ from right to left maintaining the minimum. Within the same complexities we also maintain satellite information specifying the index $k\in[i-j+2,|\PRE_j|]$ where the range minimum $\FTR[j][i]$ came from in the sub-array $\PRE_j[i-j+2\dd |\PRE_j|]$. We then construct $n$ fusion trees, one for every collection of $\ell-1$ integers in $\FTR[1\dd \ell-1][i]$. This takes total preprocessing time and space $\cO(n\ell)=\cO(n \log_\sigma n)$ by Theorem~\ref{thm:fusion}. Given the range query $[a,b]$, we need to find the smallest $j\in[1,\ell-1]$ such that $\FTR[j][b]< a$. By Theorem~\ref{thm:fusion}, we find where the predecessor of $a$ lies in $\FTR[1\dd \ell-1][b]$ in $\cO(\log_w\ell)$ time, where $w$ is the word size; this time cost is $\cO(1)$ since $w=\Theta(\log n)$.

We finally retrieve a \emph{witness} shortest absent word
as follows. If there is no $j<\ell$ such that $\FTR[j][b]<a$, then we output any shortest absent word of length $\ell$ of $T$ arbitrarily. If such a $j<\ell$ exists, by the definition of $\FTR[j][b]$, we output $T[\FTR[j][b]\dd \FTR[j][b]+j-1]$ if $\FTR[j][b]>0$ or $T[k\dd k+j-1]$ if $\FTR[j][b]=0$, where $k$ is the index of $\PRE_j$, where the minimum came from. Inspect the following illustrative example.

\begin{example}[Querying]
We construct array $\FTR$ for $T$ from Example~\ref{ex:con}. For a given $[a, b]$ we look up column $b$, and find the topmost entry whose value is less than $a$. If all entries have values greater than or equal to $a$, we output any element from $\text{SAW}_T$ arbitrarily.

\begin{table}[h]
\begin{tabular}{l|l*{18}{p{5pt}}l@{}}
     $i$&\tt1&\tt2&\tt3&\tt4&\tt5&\tt6&\tt7&\tt8&\tt9&\tt10&\tt11 &\tt12&\tt13&\tt14&\tt15&\tt16&\tt17&\tt18 \\
     $T$&\tt a&\tt b&\tt a&\tt a&\tt b&\tt a&\tt a&\tt a&\tt b&\tt b&\tt a & \tt b&\tt b&\tt b &\tt a&\tt a&\tt a&\tt b \\
     \hline
     $\FTR[1]$&\tt0&\tt1&\tt 2&\tt 2&\tt 4&\tt5&\tt 5&\tt 5&\tt 8&\tt 8&\tt 10&\tt 11&\tt 11&\tt 11&\tt 14&\tt 14&\tt 14&\tt 17\\
     $\FTR[2]$&\tt0&\tt 0&\tt0&\tt 0&\tt 0&\tt 0&\tt0&\tt 0&\tt 0&\tt5&\tt 7&\tt 7&\tt 7&\tt 7&\tt 7&\tt 11&\tt 11&\tt 13\\
     $\FTR[3]$&\tt0&\tt 0&\tt0&\tt 0&\tt 0&\tt 0&\tt 0&\tt0&\tt 0&\tt0&\tt 0&\tt 0&\tt 0&\tt 4&\tt 4&\tt4&\tt 4&\tt4\\
\end{tabular}
\end{table}

If $[a, b]= [3,14]$ then no entry in column $b=14$ is less than $a=3$, which means the length of the shortest absent word is 4; we output one from $\{\tt aaaa, \tt abab, \tt baba, \tt bbbb\}$ arbitrarily.
If $[a, b]= [5,14]$ then $\FTR[3][14]=4<5$ so the length of a shortest absent word of $T[5\dd 14]$ is 3; a shortest absent word is $T[\FTR[3][14]\dd \FTR[3][14]+3-1]=T[4\dd  6]= \tt aba$.

If $[a, b]= [7,9]$, $\FTR[2][9]=0<7$ so the length of a shortest absent word is $2$; a shortest absent word is $T[k\dd k+j-1]=T[9\dd 10]= \tt bb$ because $\FTR[2][9]=\min\{\PRE_2[9],\ldots,\PRE_2[|\PRE_2|]\}=\PRE_2[9]=0$ tells us that the minimum in this range came from index $k=9$. \end{example}

We obtain the following proposition.

\begin{proposition}\label{the:constant}
Given a string $T$ of length $n$ over an alphabet $\Sigma\subset \{1,2,\ldots,n^{\cO(1)}\}$ of size~$\sigma$, we can construct a data structure of size $\cO(n \log_{\sigma} n)$ in $\cO(n \log_{\sigma} n)$ time, so that if query $[a,b]$ is given, we can compute a shortest string over $\Sigma$ that does not occur in $T[a\dd b]$ in $\cO(1)$ time.
\end{proposition}

\section{\boldmath $\cO(n)$ Space and $\cO(\log\log_\sigma n)$ Query Time}\label{sec:loglog}

\begin{definition}[Order-$j$ Fragment]
Given a string $T$ over an alphabet of size $\sigma$ and an integer $j$, $V$ is called an \emph{order-$j$ fragment} of $T$ if and only if $V$ is a fragment of $T$ and $S_V(j)=\sigma^j$. $V$ is further called a \emph{minimal order-$j$ fragment} of $T$ if $S_U(j)<\sigma^j$ and $S_Z(j)<\sigma^j$ for $U=V[1\dd |V|-1]$ and $Z=V[2\dd |V|]$.
\end{definition}

In particular, minimal order-$j$ fragments are pairwise not included in each other. The following fact follows directly.

\begin{fact}\label{fct:V}
Given a string $T$ of length $n$ over an alphabet of size $\sigma$ and an integer $j$
we have $\cO(n)$ minimal order-$j$ fragments. Moreover, an arbitrary fragment $F$ of $T$ has $S_F[j]=\sigma^j$ if and only if it contains at least one of these minimal fragments.
\end{fact}

For each $j \in [1, \log_\sigma n]$, we consider all minimal order-$j$ fragments $T$, separately. We encode the minimal order-$j$ fragments of $T$ using two bit vectors $\OP_j$ and $\CP_j$, standing for starting positions and ending positions. Inspect the following example.

\begin{example}\label{ex:opcp}
We consider $T$ from Example~\ref{ex:con} and $j=2$.

\vspace{0.2cm}

\scalebox{0.97}{
\begin{tabular}{l|l*{25}{p{5pt}}l@{}}
     $i$&\tt1&\tt2&\tt3&\tt4&\tt5&\tt6&\tt7&\tt8&\tt9&\tt10&\tt11 &\tt12&\tt13&\tt14&\tt15&\tt16&\tt17&\tt18&\tt19&\tt20&\tt21 \\
     $T$&\tt a&\tt b&\tt a&\tt a&\tt b&\tt a&\tt a&\tt a&\tt b&\tt b&\tt a & \tt b&\tt b&\tt b &\tt a&\tt a&\tt a&\tt b \\

     \hline
     $\APP_2$&\tt 2&\tt 3&\tt 1&\tt 2&\tt 3&\tt 1&\tt 1&\tt 2&\tt 4&\tt 3&\tt 2&\tt 4&\tt 4&\tt 3&\tt 1&\tt 1&\tt 2&\tt 1&\tt 2&\tt 3&\tt 4\\
     $\PRE_2$&\tt0&\tt 0&\tt0&\tt 1&\tt 2&\tt 3&\tt 6&\tt 4&\tt 0&\tt5&\tt 8&\tt 9&\tt 12&\tt 10&\tt 7&\tt 15&\tt 11&\tt 16&\tt 17&\tt 14&\tt 13\\
     $\OP_2$&\tt 0&\tt 0&\tt 0&\tt 0&\tt 1&\tt 0&\tt 1&\tt 0&\tt 0&\tt 0&\tt 1&\tt 0&\tt 1&\tt 0&\tt 0&\tt 0&\tt 0&\tt 0\\
     $\CP_2$&\tt0&\tt 0&\tt0&\tt 0&\tt 0&\tt 0&\tt 0&\tt 0&\tt 0&\tt1&\tt 1&\tt 0&\tt 0&\tt 0&\tt 0&\tt 1&\tt 0&\tt 1\\

\end{tabular}}

\vspace{0.2cm}

For instance, $\OP_2[13]=1$ and $\CP_2[18]=1$ denote the minimal order-$2$ fragment $V=T[13\dd 18]=\texttt{bbaaab}$.
\end{example}

We construct a rank and select data structure on $\OP_j$ and $\CP_j$, for all $j \in [1, \ell-1]$ supporting $\cO(1)$-time queries. The overall space is $\cO(n)$ by Theorem~\ref{th:rs} and Fact~\ref{fct:saw}.

Let us now explain how this data structure enables fast computation of absent words of length $j$. Given a range $[a,b]$, by Fact~\ref{fct:V}, we only need to find whether $T[a\dd b]$ contains a minimal order-$j$ fragment. We can do this in $\cO(1)$ time using one rank and one select query: $t= \rank_1(a-1, \OP_j)+1$
and $\select_1(t, \CP_j)$. The select query returns the  ending position of the leftmost minimal order-$j$ fragment that starts after the position $a - 1$; it remains to check whether this minimal order-$j$ fragment is inside $[a,b]$.

\begin{example}
We consider $T$, $\OP_2$ and $\CP_2$ from Example~\ref{ex:opcp}. Let $[a,b]=[5,14]$.
We have $t= \rank_1(a-1, \OP_2)+1=\rank_1(4, \OP_2)+1=1$, $\select_1(t, \OP_2)=\select_1(1, \OP_2)= 5<b=14$ and $\select_1(t, \CP_2)=\select_1(1, \CP_2)= 10< b = 14$, which means $T[5,14]$ contains a minimal order-$2$ fragment.
\end{example}

Let us now describe a time-efficient construction of $\OP_j$ and $\CP_j$. We use arrays $\PRE_j$ and $\APP_j$ of $T$, which are constructible in $\cO(n)$ time (see Section~\ref{sec:constant}). Recall that $\PRE_j[i]$ stores the position of the rightmost occurrence of rank $\APP_j[i]$ in $\APP_j[1\dd i-1]$ (or $0$ if such an occurrence does not exist). We apply Fact~\ref{fac:trick} as follows.
We start with all bits of $\OP_j$ and $\CP_j$ unset.
Then, for each $b \in [1,n]$ for which $\PRE_j[b-j+1]<\min\{\PRE_j[i]:i\in[b-j+2,|\PRE_j|]\} = a$, we set the $b$th bit of $\CP_j$ and the $a$th bit of $\OP_j$.
This can be done online in a right-to-left scan of $\PRE_j$ in $\cO(n)$ time.

\begin{example}
We consider $T$, $\OP_2$ and $\CP_2$ from Example~\ref{ex:opcp}. We start by setting $b=n=18$ and scan $\PRE_2$ from right to left: we have $a=13$ because
$\min\{\PRE_2[i]:i\in[18,21]\} = 13$.
This gives fragment $T[13\dd 18]$, which is minimal since $\PRE_2[b-1] = \PRE_2[17] < 13$.
Then we set $b=n-1=17$ and have $a=11$ because $\min\{\PRE_2[i]:i\in[17,21]\} = 11$.
This gives fragment $T[11\dd 17]$, which is not minimal since $\PRE_2[b-1] = \PRE_2[16] \geq 11$.
Then we set $b=n-2=16$ and have $a=11$ because $\min\{\PRE_2[i]:i\in[16,21]\}=11$.
This gives fragment $T[11\dd 16]$, which is minimal since $\PRE_2[b-1] = \PRE_2[15] < 11$ .
\end{example}

\begin{lemma}\label{lem:cpop}
$\OP_j$ and $\CP_j$ can be constructed in $\cO(n)$ time.
\end{lemma}

For all $j$, the construction time is $\cO(n\ell)=\cO(n\log_\sigma n)$ by Theorem~\ref{th:rs}, Lemma~\ref{lem:cpop}, and Fact~\ref{fct:saw}. All the arrays $\OP_j$ and $\CP_j$ in total occupy $\cO(n\ell) = \cO(n\log_\sigma n)$ \emph{bits} of space, which is $\cO(n)$ space when measured in $\Theta(\log n)$-bit machine words. We obtain the following lemma.

\begin{lemma}\label{lem:fixedj}
Given a string $T$ of length $n$ over an alphabet $\Sigma\subset \{1,2,\ldots,n^{\cO(1)}\}$ of size $\sigma$, we can construct a data structure of size $\cO(n)$ in $\cO(n \log_{\sigma} n)$ time, so that if query $(j,[a,b])$ is given, we can check in $\cO(1)$ time whether there is any string in $\Sigma^j$ that does not occur in $T[a\dd b]$, and if so return such a string.
\end{lemma}

We can now perform binary search on $j$ using Lemma~\ref{lem:fixedj} to find the smallest $j$ for which $S_{T[a\dd b]}(j)<\sigma^j$. This results in $\cO(\log\ell)=\cO(\log\log_\sigma n)$ query time by Fact~\ref{fct:saw}. It should now be clear that when we find the $j$ corresponding to the length of a shortest absent word, we can output the length-$j$ suffix of the leftmost minimal order-$j$ fragment starting after $a$.
Note that outputting this suffix is correct by the definition of minimal order-$j$ fragments.

\begin{example}
We consider $T$, $\OP_2$ and $\CP_2$ from Example~\ref{ex:opcp}. Let $[a,b]=[2,7]$. The length of a shortest absent word of $T[2\dd 7]$ is $2$. We output $\texttt{bb}$, which is the length-$2$ suffix of the leftmost minimal order-$2$ fragment $T[5\dd 10]=\texttt{baaabb}$ starting after $a=2$.
\end{example}

We obtain the following result.

\begin{proposition}\label{the:loglog}
Given a string $T$ of length $n$ over an alphabet $\Sigma\subset \{1,2,\ldots,n^{\cO(1)}\}$ of size~$\sigma$, we can construct a data structure of size $\cO(n)$ in $\cO(n \log_{\sigma} n)$ time, so that if query $[a,b]$ is given, we can compute a shortest string over $\Sigma$ that does not occur in $T[a\dd b]$ in $\cO(\log\log_\sigma n)$ time.
\end{proposition}

\section{Combinatorial Insights}\label{sec:comb}

A positive integer $p$ is a \emph{period} of a string $S$ if $S[i] = S[i + p]$ for all $i \in [1, |S| - p]$.  We refer to the smallest period as \emph{the period} of the string. Let us state the periodicity lemma, one of the most elegant combinatorial results on strings.

\begin{lemma}[Periodicity Lemma (weak version)~\cite{FineWilf}] If a string $S$ has periods $p$ and $q$ such that $p+q \leq |S|$, then $\textsf{gcd}(p, q)$ is also a period of $S$.\label{lem:per}
\end{lemma}

\begin{lemma}\label{lem:long}
If all strings in $\{UW: U \in \Sigma^k\}$ for $W\neq \varepsilon$ occur in some string $S$, then $|S|\geq |W| \cdot \sigma^k/4$.
\end{lemma}
\begin{proof}
Let $p$ be the period of $W$, and let $a \in \Sigma$ be such that the period of $aW$ is also~$p$.
All strings $Z b W$ for a letter $b\neq a$ and $Z \in \Sigma^{k-1}$ must occur in $S$.
Let $A=\{UW : U\in\Sigma^{k}\} \setminus \{ZaW : Z\in \Sigma^{k-1}\}$, and note that it is of size $\sigma^{k}-\sigma^{k-1}\geq \sigma^{k}/2$.
The following claim immediately implies the statement of the lemma.

\begin{claim*}
Let $i$ and $j$ be starting positions of occurrences of different strings $UW, VW \in A$ in~$S$, respectively.
Then, we have $|j-i| \geq  |W|/2$.
\end{claim*}
\begin{proof}
Let us assume, without loss of generality, that $j>i$.
Further, let us assume towards a contradiction that $j-i <|W|/2$.
Then, $j-i$ is a period of $W$ and $p+j-i \leq |W|$ since $p\leq j-i$. Therefore, due to the periodicity lemma (Lemma~\ref{lem:per}), $j-i$ must be divisible by the period $p$ of $W$.
Hence, $V$ ends with the letter $a$ and $VW \notin A$, a contradiction.
\end{proof}
This concludes the proof of this lemma.
\end{proof}

\begin{lemma}\label{lem:ext}
If a shortest absent word of a string $X$ is of length $\lambda$, then the length of a shortest absent word of $XY$ is in $[\lambda, \lambda + \max\{10,4+\log_{\sigma} (|Y|/\lambda)\}]$.
\end{lemma}
\begin{proof}
Let $W$ and $W'$ be shortest absent words of $X$ and $XY$, respectively. Further, let $d=|W'| -|W|$.  In order to have  $d>0$, all strings $UW$ for $U \in \Sigma^{d-1}$ must occur in $XY$, and hence in $X[|X|-|UW|+2 \dd |X|]\cdot Y$, since none of them occurs in $X$.
Lemma~\ref{lem:long} implies that $|Y|+\lambda+d > \lambda \cdot \sigma^{d-1}/4$.
Then, since $\lambda+d \leq 2\lambda d$ for any positive integers $\lambda, d$, we have $|Y|>\lambda \cdot(\sigma^{d-1}/4 -2d)$.
Assuming that $d\geq 10$, and since $\sigma \geq 2$, we conclude that $|Y|>\lambda \cdot \sigma^{d-1}/8$.
Consequently, $\log_{\sigma} (8|Y|/\lambda) +1> d$.
Since $\log_\sigma 8 \leq 3$ we get the claimed bound.
\end{proof}

\begin{lemma}\label{lem:rev}
If a shortest absent word of $XY$ is of length $m$, a shortest absent word of $X$ is of length $\lambda$,
and $|Y|\leq m \cdot \tau$, for a positive integer $\tau\geq 16$, then $m-\lambda \leq 10 + 2\log_{\sigma} \tau$.
\end{lemma}
\begin{proof}
From \cref{lem:ext} we have $\lambda \in [m - \max\{10,4+\log_{\sigma} (|Y|/\lambda)\}, m]$.
If $\max\{10,4+\log_{\sigma} (|Y|/\lambda)\}=10$, then $m-\lambda \leq 10$ and we are done.

In the complementary case, since $|Y|\leq m \cdot \tau$, we get the following:
\[\lambda\geq m-\log_{\sigma}(m \cdot \tau/\lambda)-4\iff
\lambda \geq m + \log_{\sigma}\lambda-\log_{\sigma}m -\log_{\sigma}\tau -4.\]

In particular, $\lambda\geq m-\log_{\sigma}m -\log_{\sigma}\tau-4$.

From the above, if $m\leq \tau$, then $m-\lambda\leq 4 + 2 \log_{\sigma}\tau$.

In what follows we assume that $m>\tau\geq 16$.
Rearranging the original equation, and since $\log_\sigma(\cdot)$  is an increasing function and $\lambda \geq m-\log_{\sigma}m-\log_{\sigma}\tau -4$, we have
\begin{multline*}
m-\lambda \leq 4+\log_{\sigma}(m \cdot \tau/\lambda) \leq  4+\log_{\sigma}\left(\frac{m}{m-\log_{\sigma}m-\log_{\sigma}\tau-4}\right)+ \log_{\sigma}\tau\\
\leq 4+\log_{\sigma}\left(\frac{m}{m-2\log_{\sigma}m-4}\right)+ \log_{\sigma}\tau.
\end{multline*}
Then, we have $m-2\log_{\sigma}m-4\geq m/5$ since, for any $\sigma \geq 2$, $4x/5-2\log_{\sigma}x-4$ is an increasing function on $[16,\infty)$ and positive for $x=16$.
Hence, $m-\lambda \leq 4 + \log_\sigma 5+\log_\sigma \tau\leq 7 +\log_\sigma \tau$.

By combining the bounds on $m-\lambda$ we get the claimed bound.
\end{proof}

\section{\boldmath $\cO(n)$ Space and $\cO(1)$ Query Time}\label{sec:main}

Our linear-space solution of the \ISAW{} problem with constant query time is an optimization of the $\cO(n\log_\sigma n)$-space solution from Section~\ref{sec:constant} with some ``boundary'' cases processed using the data structure of Section~\ref{sec:loglog}. Let us first describe a simpler $\cO(n \log_\sigma \log n)$-space data structure, which will be then optimized using the combinatorial insights from Section~\ref{sec:comb}.

Recall that we denote by $\ell$ the length of a shortest absent word of $T$. The issue with the solution of Section~\ref{sec:constant} is that the 2d array $\FTR[1\dd \ell-1][1\dd n]$, equipped with fusion trees, occupies $\cO(n \log_\sigma n)$ space. In order to reduce the memory consumption, we store the array $\FTR$ implicitly, computing its entries on demand, and utilize succinct fusion trees from Theorem~\ref{thm:fusion-succinct} instead of usual fusion trees.

Recall that $\FTR[j][i]$ is the rightmost index of $T$ such that $T[\FTR[j][i]\dd i]$ contains as substrings all strings of length $j$ over $\Sigma$ and it is equal to $\min\{\PRE_j[i-j+2],\ldots,\PRE_j[|\PRE_j|]\}$. Therefore, the content of the 2d array $\FTR$ can be ``emulated'' without storing it explicitly if one can compute in $\cO(1)$ time the minima $\min\{\PRE_j[a], \ldots,\PRE_j[|\PRE_j|]\}$, for any $a \in [1,n]$. For $j \in [1,\ell-1]$ and $a \in [1,n]$, denote $M_{j,a} = \min\{\PRE_j[a], \ldots,\PRE_j[|\PRE_j|]\}$. Let us fix some $j$. Since the sequence $M_{j,1}, M_{j,2}, \ldots, M_{j,n}$ is non-decreasing, we can encode it in a $2n$-bit array $B_j$ using the select data structure from Theorem~\ref{th:rs} as follows: we construct $B_j$ (initially empty) by considering $a = 1,2,\ldots,n$ in increasing order and, for each $a$, we append to the end of $B_j$ exactly $M_{j,a} - M_{j,a-1}$ zeroes followed by~1, setting $M_{j,0}=0$ (i.e., we append the number $M_{j,a} - M_{j,a-1}$ written in unary); then, we have $M_{j,a} = \select_1(a, B_j) - a$.

\begin{example}
We consider $T$ from Example~\ref{ex:con} and $j=2$.

\vspace{0.2cm}

\scalebox{0.97}{
\begin{tabular}{l|l*{25}{p{5pt}}l@{}} 
     $i$&\tt1&\tt2&\tt3&\tt4&\tt5&\tt6&\tt7&\tt8&\tt9&\tt10&\tt11 &\tt12&\tt13&\tt14&\tt15&\tt16&\tt17&\tt18&\tt19&\tt20&\tt21 \\
     $T$&\tt a&\tt b&\tt a&\tt a&\tt b&\tt a&\tt a&\tt a&\tt b&\tt b&\tt a & \tt b&\tt b&\tt b &\tt a&\tt a&\tt a&\tt b \\

     \hline
     $\APP_2$&\tt 2&\tt 3&\tt 1&\tt 2&\tt 3&\tt 1&\tt 1&\tt 2&\tt 4&\tt 3&\tt 2&\tt 4&\tt 4&\tt 3&\tt 1&\tt 1&\tt 2&\tt 1&\tt 2&\tt 3&\tt 4\\
     $\PRE_2$&\tt 0&\tt 0&\tt 0&\tt 1&\tt 2&\tt 3&\tt 6&\tt 4&\tt 0&\tt 5&\tt 8&\tt 9&\tt 12&\tt 10&\tt 7&\tt 15&\tt 11&\tt 16&\tt 17&\tt 14&\tt 13\\
     $M_{2,i}$&\tt 0&\tt 0&\tt 0&\tt 0&\tt 0&\tt 0&\tt 0&\tt 0&\tt 0&\tt 5&\tt 7&\tt 7&\tt 7&\tt 7&\tt 7&\tt 11&\tt 11&\tt 13
\end{tabular}}

\vspace{0.2cm}

In this case, we have $B_{2} = \tt 1\tt 1\tt 1\tt 1\tt 1\tt 1\tt 1\tt 1\tt 1\tt 0\tt 0\tt 0\tt 0\tt 0\tt 1\tt 0\tt 0\tt 1\tt 1\tt 1\tt 1\tt 1\tt 0\tt 0\tt 0\tt 0\tt 1\tt 1\tt 0\tt 0\tt 1$.
\end{example}

Besides access to the 2d array $\FTR$, the algorithm of Section~\ref{sec:constant} also required access to the values $\argmin\{\PRE_j[a], \ldots,\PRE_j[|\PRE_j|]\}$ in order to retrieve a witness shortest absent word. To this end, we build the $2n$-bit RMQ data structure from Theorem~\ref{th:rmq} on each array $\PRE_j$; the data structure does not need to store the array $\PRE_j$ itself to compute $\argmin$. The arrays $B_j$, for $j \in [1,\ell-1]$, equipped with select data structures, and the RMQ data structures on arrays $\PRE_j$, for $j \in [1,\ell-1]$, can be constructed in total $\cO(n \ell) = \cO(n\log_\sigma n)$ time and they altogether occupy $\cO(n\log_\sigma n)$ \emph{bits} of space, which is $\cO(n)$ space when measured in machine words.

To answer a query $[a,b]$, it suffices to find the smallest $j$ such that $\FTR[j][b]<a$. We do this by finding where the predecessor of $a$ lies in $\FTR[1\dd \ell-1][b]$.
To this end, we constructed $n$ fusion trees: one per $\FTR[1 \dd \ell-1][i]$, resulting in a data structure of size $\Theta(n\ell)=\cO(n\log_\sigma n)$ with $\cO(1)$ query time. But now we do not store the arrays $\FTR[1 \dd \ell-1][i]$ explicitly, while still having $\cO(1)$-time ``oracle'' access to their entries on demand. Hence, we can construct a succinct fusion tree of Theorem~\ref{thm:fusion-succinct}, for each array $\FTR[1 \dd \ell-1][i]$, which takes $\cO(\ell \log\log n)$ \emph{bits} of space since the size of machine words is $w = \Theta(\log n)$ bits (a shared table mentioned in Theorem~\ref{thm:fusion-succinct} is also precomputed for all the trees in $o(2^w) = o(n)$ time).

Thus, all the succinct fusion trees can be constructed in $\cO(n\log_\sigma n)$ time and occupy $\cO(n\log_\sigma n\log\log n)$ \emph{bits}, which is $\cO(n\log_\sigma\log n)$ space when measured in $\Theta(\log n)$-bit machine words. The \ISAW{} queries are answered in $\cO(1)$ time by the same algorithm as in Section~\ref{sec:constant}.

\medskip

Now we are to further reduce the memory usage of the data structure. We truncate all the arrays $\FTR[0 \dd \ell-1][i]$ except those where $i$ is a multiple of $\lfloor\log\log n\rfloor$ or $i=n$: namely, if $i$ is a multiple of $\lfloor\log\log n\rfloor$ or $i=n$, then the succinct fusion tree for the whole array $\FTR[1 \dd \ell-1][i]$ is stored, occupying $\cO(\ell \log\log n)$ bits, by Theorem~\ref{thm:fusion-succinct}; otherwise ($i\neq n$ is not a multiple of $\lfloor \log\log n\rfloor$), we store the succinct fusion tree only for the subarray $\FTR[1 \dd \lceil \log n / \log\log n\rceil][i]$, thus taking $\cO(\log n)$ bits,  by Theorem~\ref{thm:fusion-succinct}. In total, the space used is $\cO(\frac{n}{\log\log n}\ell\log\log n  + n\log n) = \cO(n\log n)$ in bits or $\cO(n)$ in words.

In order to answer an \ISAW{} query for $T[a \dd b]$, we first check whether the length $\lambda$ of a shortest absent word in $T[a \dd b]$ is smaller than $\log n / \log\log n$ by querying the fusion tree of $\FTR[1 \dd \lceil\log n / \log\log n\rceil][b]$. If it is the case, then we have computed the length $\lambda$ and we find the absent word itself using RMQs exactly as in the $\cO(n\log_\sigma\log n)$-space solution described above.

Suppose that $\lambda \geq \log n / \log\log n$.
We compute $b'$, the successor of $b$ among the positions $i$ for which we have not truncated $\FTR[1 \dd \ell-1][i]$: $b' = \min\{n,\lceil b / \lfloor\log\log n\rfloor\rceil \cdot \lfloor\log\log n\rfloor\}$. Observe that $[a,b] \subseteq [a,b']$. Then, using the fusion tree of $\FTR[1 \dd \ell-1][b']$, we compute the smallest $m$ such that $\FTR[m][b'] < a$.
Then, $m$ is the length of a shortest absent word in $T[a \dd b']$. Denote $X = T[a \dd b]$ and $T[a \dd b'] = XY$ where $Y$ is a suffix of $T[a\dd b']$ of length $b' - b$. We obviously have $\lambda \leq m$. Since $|Y| = b' - b < \log\log n$ and $m \geq \log n / \log\log n$, we have $|Y| < m$. It follows from Lemma~\ref{lem:rev} that the answer $\lambda$ is within a range of length $18$ from $m$. Therefore, $\lambda$ belongs to the range $[m-18, m]$ and we can find it in $\cO(1)$ time using $\cO(1)$ queries of the $\cO(n)$-space data structure encapsulated by~\cref{lem:fixedj}. We thus arrive at the main result of the paper.

\mainthm*

\section{Open Problems}\label{sec:fin}

It remains open whether a data structure for the \ISAW{} problem with the same query time and space complexities as the one encapsulated in Theorem~\ref{the:main} can be constructed in linear time. Also, it is natural to pose the following related open problem, which may require the development of fundamentally different techniques. Given a string $T$ of length $n$ over an alphabet $\Sigma\subset \{1,2,\ldots,n^{\cO(1)}\}$, preprocess $T$ so that given a range $[i,j]$, we can return a representation of a shortest string over $\Sigma_{[i,j]}$ that is absent in the fragment $T[i]\cdots T[j]$ of~$T$, where $\Sigma_{[i,j]}$ is the set of letters from $\Sigma$ occurring in the fragment $T[i]\cdots T[j]$.

\bibliography{references}

\begin{thebibliography}{59}
\expandafter\ifx\csname natexlab\endcsname\relax\def\natexlab#1{#1}\fi
\providecommand{\url}[1]{\texttt{#1}}
\providecommand{\href}[2]{#2}
\providecommand{\path}[1]{#1}
\providecommand{\DOIprefix}{doi:}
\providecommand{\ArXivprefix}{arXiv:}
\providecommand{\URLprefix}{URL: }
\providecommand{\Pubmedprefix}{pmid:}
\providecommand{\doi}[1]{\href{http://dx.doi.org/#1}{\path{#1}}}
\providecommand{\Pubmed}[1]{\href{pmid:#1}{\path{#1}}}
\providecommand{\bibinfo}[2]{#2}
\ifx\xfnm\relax \def\xfnm[#1]{\unskip,\space#1}\fi
\bibitem[{Abedin et~al.(2018)Abedin, Ganguly, Hon, Nekrich, Sadakane, Shah and
  Thankachan}]{DBLP:conf/cocoon/Abedin0HNSST18}
\bibinfo{author}{Abedin, P.}, \bibinfo{author}{Ganguly, A.},
  \bibinfo{author}{Hon, W.}, \bibinfo{author}{Nekrich, Y.},
  \bibinfo{author}{Sadakane, K.}, \bibinfo{author}{Shah, R.},
  \bibinfo{author}{Thankachan, S.V.}, \bibinfo{year}{2018}.
\newblock \bibinfo{title}{A linear-space data structure for {Range-LCP} queries
  in poly-logarithmic time}, in: \bibinfo{booktitle}{Computing and
  Combinatorics - 24th International Conference, {COCOON} 2018}, pp.
  \bibinfo{pages}{615--625}.
\newblock \URLprefix \url{https://doi.org/10.1007/978-3-319-94776-1\_51},
  \DOIprefix\doi{10.1007/978-3-319-94776-1\_51}.
\bibitem[{Abedin et~al.(2020)Abedin, Ganguly, Pissis and
  Thankachan}]{DBLP:journals/algorithms/Abedin0PT20}
\bibinfo{author}{Abedin, P.}, \bibinfo{author}{Ganguly, A.},
  \bibinfo{author}{Pissis, S.P.}, \bibinfo{author}{Thankachan, S.V.},
  \bibinfo{year}{2020}.
\newblock \bibinfo{title}{Efficient data structures for range shortest unique
  substring queries}.
\newblock \bibinfo{journal}{Algorithms} \bibinfo{volume}{13},
  \bibinfo{pages}{276}.
\newblock \URLprefix \url{https://doi.org/10.3390/a13110276},
  \DOIprefix\doi{10.3390/a13110276}.
\bibitem[{Amir et~al.(2014)Amir, Apostolico, Landau, Levy, Lewenstein and
  Porat}]{DBLP:journals/jcss/AmirALLLP14}
\bibinfo{author}{Amir, A.}, \bibinfo{author}{Apostolico, A.},
  \bibinfo{author}{Landau, G.M.}, \bibinfo{author}{Levy, A.},
  \bibinfo{author}{Lewenstein, M.}, \bibinfo{author}{Porat, E.},
  \bibinfo{year}{2014}.
\newblock \bibinfo{title}{Range {LCP}}.
\newblock \bibinfo{journal}{J. Comput. Syst. Sci.} \bibinfo{volume}{80},
  \bibinfo{pages}{1245--1253}.
\newblock \URLprefix \url{https://doi.org/10.1016/j.jcss.2014.02.010},
  \DOIprefix\doi{10.1016/j.jcss.2014.02.010}.
\bibitem[{Amir et~al.(2020)Amir, Charalampopoulos, Pissis and
  Radoszewski}]{DBLP:journals/algorithmica/AmirCPR20}
\bibinfo{author}{Amir, A.}, \bibinfo{author}{Charalampopoulos, P.},
  \bibinfo{author}{Pissis, S.P.}, \bibinfo{author}{Radoszewski, J.},
  \bibinfo{year}{2020}.
\newblock \bibinfo{title}{Dynamic and internal longest common substring}.
\newblock \bibinfo{journal}{Algorithmica} \bibinfo{volume}{82},
  \bibinfo{pages}{3707--3743}.
\newblock \URLprefix \url{https://doi.org/10.1007/s00453-020-00744-0},
  \DOIprefix\doi{10.1007/s00453-020-00744-0}.
\bibitem[{Ayad et~al.(2019)Ayad, Badkobeh, Fici, H{\'{e}}liou and
  Pissis}]{dcc19}
\bibinfo{author}{Ayad, L.A.K.}, \bibinfo{author}{Badkobeh, G.},
  \bibinfo{author}{Fici, G.}, \bibinfo{author}{H{\'{e}}liou, A.},
  \bibinfo{author}{Pissis, S.P.}, \bibinfo{year}{2019}.
\newblock \bibinfo{title}{Constructing antidictionaries in output-sensitive
  space}, in: \bibinfo{booktitle}{Data Compression Conference, {DCC} 2019},
  \bibinfo{publisher}{{IEEE}}. pp. \bibinfo{pages}{538--547}.
\newblock \URLprefix \url{https://doi.org/10.1109/DCC.2019.00062},
  \DOIprefix\doi{10.1109/DCC.2019.00062}.
\bibitem[{Babenko et~al.(2016)Babenko, Gawrychowski, Kociumaka, Kolesnichenko
  and Starikovskaya}]{DBLP:journals/tcs/BabenkoGKKS16}
\bibinfo{author}{Babenko, M.A.}, \bibinfo{author}{Gawrychowski, P.},
  \bibinfo{author}{Kociumaka, T.}, \bibinfo{author}{Kolesnichenko, I.I.},
  \bibinfo{author}{Starikovskaya, T.}, \bibinfo{year}{2016}.
\newblock \bibinfo{title}{Computing minimal and maximal suffixes of a
  substring}.
\newblock \bibinfo{journal}{Theor. Comput. Sci.} \bibinfo{volume}{638},
  \bibinfo{pages}{112--121}.
\newblock \URLprefix \url{https://doi.org/10.1016/j.tcs.2015.08.023},
  \DOIprefix\doi{10.1016/j.tcs.2015.08.023}.
\bibitem[{Babenko et~al.(2015)Babenko, Gawrychowski, Kociumaka and
  Starikovskaya}]{DBLP:conf/soda/BabenkoGKS15}
\bibinfo{author}{Babenko, M.A.}, \bibinfo{author}{Gawrychowski, P.},
  \bibinfo{author}{Kociumaka, T.}, \bibinfo{author}{Starikovskaya, T.},
  \bibinfo{year}{2015}.
\newblock \bibinfo{title}{Wavelet trees meet suffix trees}, in:
  \bibinfo{booktitle}{Proceedings of the Twenty-Sixth Annual {ACM-SIAM}
  Symposium on Discrete Algorithms, {SODA} 2015}, \bibinfo{publisher}{{SIAM}}.
  pp. \bibinfo{pages}{572--591}.
\newblock \URLprefix \url{https://doi.org/10.1137/1.9781611973730.39},
  \DOIprefix\doi{10.1137/1.9781611973730.39}.
\bibitem[{Badkobeh et~al.(2021)Badkobeh, Charalampopoulos and
  Pissis}]{self:CPM}
\bibinfo{author}{Badkobeh, G.}, \bibinfo{author}{Charalampopoulos, P.},
  \bibinfo{author}{Pissis, S.P.}, \bibinfo{year}{2021}.
\newblock \bibinfo{title}{Internal shortest absent word queries}, in:
  \bibinfo{editor}{Gawrychowski, P.}, \bibinfo{editor}{Starikovskaya, T.}
  (Eds.), \bibinfo{booktitle}{32nd Annual Symposium on Combinatorial Pattern
  Matching, {CPM} 2021}, \bibinfo{publisher}{Schloss Dagstuhl--Leibniz-Zentrum
  fuer Informatik}, \bibinfo{address}{Dagstuhl, Germany}. pp.
  \bibinfo{pages}{24:1--24:18}.
\bibitem[{Barton et~al.(2014)Barton, H{\'{e}}liou, Mouchard and Pissis}]{MAW}
\bibinfo{author}{Barton, C.}, \bibinfo{author}{H{\'{e}}liou, A.},
  \bibinfo{author}{Mouchard, L.}, \bibinfo{author}{Pissis, S.P.},
  \bibinfo{year}{2014}.
\newblock \bibinfo{title}{Linear-time computation of minimal absent words using
  suffix array}.
\newblock \bibinfo{journal}{{BMC} Bioinform.} \bibinfo{volume}{15},
  \bibinfo{pages}{388}.
\newblock \URLprefix \url{https://doi.org/10.1186/s12859-014-0388-9},
  \DOIprefix\doi{10.1186/s12859-014-0388-9}.
\bibitem[{Barton et~al.(2015)Barton, H{\'{e}}liou, Mouchard and
  Pissis}]{DBLP:conf/ppam/BartonHMP15}
\bibinfo{author}{Barton, C.}, \bibinfo{author}{H{\'{e}}liou, A.},
  \bibinfo{author}{Mouchard, L.}, \bibinfo{author}{Pissis, S.P.},
  \bibinfo{year}{2015}.
\newblock \bibinfo{title}{Parallelising the computation of minimal absent
  words}, in: \bibinfo{booktitle}{Parallel Processing and Applied Mathematics -
  11th International Conference, {PPAM} 2015. Revised Selected Papers, Part
  {II}}, \bibinfo{publisher}{Springer}. pp. \bibinfo{pages}{243--253}.
\newblock \URLprefix \url{https://doi.org/10.1007/978-3-319-32152-3\_23},
  \DOIprefix\doi{10.1007/978-3-319-32152-3\_23}.
\bibitem[{Bender and Farach{-}Colton(2000)}]{DBLP:conf/latin/BenderF00}
\bibinfo{author}{Bender, M.A.}, \bibinfo{author}{Farach{-}Colton, M.},
  \bibinfo{year}{2000}.
\newblock \bibinfo{title}{The {LCA} problem revisited}, in:
  \bibinfo{booktitle}{{LATIN} 2000: Theoretical Informatics, 4th Latin American
  Symposium, Proceedings}, \bibinfo{publisher}{Springer}. pp.
  \bibinfo{pages}{88--94}.
\newblock \URLprefix \url{https://doi.org/10.1007/10719839\_9},
  \DOIprefix\doi{10.1007/10719839\_9}.
\bibitem[{Berkman and Vishkin(1993)}]{DBLP:journals/siamcomp/BerkmanV93}
\bibinfo{author}{Berkman, O.}, \bibinfo{author}{Vishkin, U.},
  \bibinfo{year}{1993}.
\newblock \bibinfo{title}{Recursive star-tree parallel data structure}.
\newblock \bibinfo{journal}{{SIAM} J. Comput.} \bibinfo{volume}{22},
  \bibinfo{pages}{221--242}.
\newblock \URLprefix \url{https://doi.org/10.1137/0222017},
  \DOIprefix\doi{10.1137/0222017}.
\bibitem[{Birenzwige et~al.(2020)Birenzwige, Golan and
  Porat}]{DBLP:conf/soda/BirenzwigeGP20}
\bibinfo{author}{Birenzwige, O.}, \bibinfo{author}{Golan, S.},
  \bibinfo{author}{Porat, E.}, \bibinfo{year}{2020}.
\newblock \bibinfo{title}{Locally consistent parsing for text indexing in small
  space}, in: \bibinfo{booktitle}{Proceedings of the 2020 {ACM-SIAM} Symposium
  on Discrete Algorithms, {SODA} 2020}, \bibinfo{publisher}{{SIAM}}. pp.
  \bibinfo{pages}{607--626}.
\newblock \URLprefix \url{https://doi.org/10.1137/1.9781611975994.37},
  \DOIprefix\doi{10.1137/1.9781611975994.37}.
\bibitem[{Chairungsee and Crochemore(2012)}]{Chairungsee2012109}
\bibinfo{author}{Chairungsee, S.}, \bibinfo{author}{Crochemore, M.},
  \bibinfo{year}{2012}.
\newblock \bibinfo{title}{Using minimal absent words to build phylogeny}.
\newblock \bibinfo{journal}{Theor. Comput. Sci.} \bibinfo{volume}{450},
  \bibinfo{pages}{109--116}.
\newblock \URLprefix \url{https://doi.org/10.1016/j.tcs.2012.04.031},
  \DOIprefix\doi{10.1016/j.tcs.2012.04.031}.
\bibitem[{Chan et~al.(2011)Chan, Larsen and
  P{\u{a}}tra{\c{s}}cu}]{DBLP:conf/compgeom/ChanLP11}
\bibinfo{author}{Chan, T.M.}, \bibinfo{author}{Larsen, K.G.},
  \bibinfo{author}{P{\u{a}}tra{\c{s}}cu, M.}, \bibinfo{year}{2011}.
\newblock \bibinfo{title}{Orthogonal range searching on the ram, revisited},
  in: \bibinfo{editor}{Hurtado, F.}, \bibinfo{editor}{van Kreveld, M.J.}
  (Eds.), \bibinfo{booktitle}{Proceedings of the 27th {ACM} Symposium on
  Computational Geometry, Paris, France, June 13-15, 2011},
  \bibinfo{publisher}{{ACM}}. pp. \bibinfo{pages}{1--10}.
\newblock \URLprefix \url{https://doi.org/10.1145/1998196.1998198},
  \DOIprefix\doi{10.1145/1998196.1998198}.
\bibitem[{Charalampopoulos et~al.(2018a)Charalampopoulos, Crochemore, Fici,
  Merca\c{s} and Pissis}]{IC18}
\bibinfo{author}{Charalampopoulos, P.}, \bibinfo{author}{Crochemore, M.},
  \bibinfo{author}{Fici, G.}, \bibinfo{author}{Merca\c{s}, R.},
  \bibinfo{author}{Pissis, S.P.}, \bibinfo{year}{2018}a.
\newblock \bibinfo{title}{Alignment-free sequence comparison using absent
  words}.
\newblock \bibinfo{journal}{Inf. Comput.} \bibinfo{volume}{262},
  \bibinfo{pages}{57--68}.
\newblock \URLprefix \url{https://doi.org/10.1016/j.ic.2018.06.002},
  \DOIprefix\doi{10.1016/j.ic.2018.06.002}.
\bibitem[{Charalampopoulos et~al.(2018b)Charalampopoulos, Crochemore and
  Pissis}]{DBLP:conf/spire/Charalampopoulos18}
\bibinfo{author}{Charalampopoulos, P.}, \bibinfo{author}{Crochemore, M.},
  \bibinfo{author}{Pissis, S.P.}, \bibinfo{year}{2018}b.
\newblock \bibinfo{title}{On extended special factors of a word}, in:
  \bibinfo{booktitle}{String Processing and Information Retrieval - 25th
  International Symposium, {SPIRE} 2018}, \bibinfo{publisher}{Springer}. pp.
  \bibinfo{pages}{131--138}.
\newblock \URLprefix \url{https://doi.org/10.1007/978-3-030-00479-8\_11},
  \DOIprefix\doi{10.1007/978-3-030-00479-8\_11}.
\bibitem[{Charalampopoulos et~al.(2021)Charalampopoulos, Gawrychowski, Mozes
  and Weimann}]{DBLP:journals/corr/abs-2103-03294}
\bibinfo{author}{Charalampopoulos, P.}, \bibinfo{author}{Gawrychowski, P.},
  \bibinfo{author}{Mozes, S.}, \bibinfo{author}{Weimann, O.},
  \bibinfo{year}{2021}.
\newblock \bibinfo{title}{An almost optimal edit distance oracle}.
\newblock \bibinfo{journal}{CoRR} \bibinfo{volume}{abs/2103.03294}.
\newblock \href{http://arxiv.org/abs/2103.03294}{{\tt arXiv:2103.03294}}.
\bibitem[{Charalampopoulos et~al.(2020a)Charalampopoulos, Kociumaka, Mohamed,
  Radoszewski, Rytter, Straszynski, Walen and
  Zuba}]{DBLP:conf/cpm/Charalampopoulos20}
\bibinfo{author}{Charalampopoulos, P.}, \bibinfo{author}{Kociumaka, T.},
  \bibinfo{author}{Mohamed, M.}, \bibinfo{author}{Radoszewski, J.},
  \bibinfo{author}{Rytter, W.}, \bibinfo{author}{Straszynski, J.},
  \bibinfo{author}{Walen, T.}, \bibinfo{author}{Zuba, W.},
  \bibinfo{year}{2020}a.
\newblock \bibinfo{title}{Counting distinct patterns in internal dictionary
  matching}, in: \bibinfo{booktitle}{31st Annual Symposium on Combinatorial
  Pattern Matching, {CPM} 2020}, \bibinfo{publisher}{Schloss Dagstuhl -
  Leibniz-Zentrum f{\"{u}}r Informatik}. pp. \bibinfo{pages}{8:1--8:15}.
\newblock \URLprefix \url{https://doi.org/10.4230/LIPIcs.CPM.2020.8},
  \DOIprefix\doi{10.4230/LIPIcs.CPM.2020.8}.
\bibitem[{Charalampopoulos et~al.(2019)Charalampopoulos, Kociumaka, Mohamed,
  Radoszewski, Rytter and Walen}]{DBLP:conf/isaac/Charalampopoulos19}
\bibinfo{author}{Charalampopoulos, P.}, \bibinfo{author}{Kociumaka, T.},
  \bibinfo{author}{Mohamed, M.}, \bibinfo{author}{Radoszewski, J.},
  \bibinfo{author}{Rytter, W.}, \bibinfo{author}{Walen, T.},
  \bibinfo{year}{2019}.
\newblock \bibinfo{title}{Internal dictionary matching}, in:
  \bibinfo{booktitle}{30th International Symposium on Algorithms and
  Computation, {ISAAC} 2019}, \bibinfo{publisher}{Schloss Dagstuhl -
  Leibniz-Zentrum f{\"{u}}r Informatik}. pp. \bibinfo{pages}{22:1--22:17}.
\newblock \URLprefix \url{https://doi.org/10.4230/LIPIcs.ISAAC.2019.22},
  \DOIprefix\doi{10.4230/LIPIcs.ISAAC.2019.22}.
\bibitem[{Charalampopoulos et~al.(2020b)Charalampopoulos, Kociumaka and
  Wellnitz}]{DBLP:conf/focs/Charalampopoulos20}
\bibinfo{author}{Charalampopoulos, P.}, \bibinfo{author}{Kociumaka, T.},
  \bibinfo{author}{Wellnitz, P.}, \bibinfo{year}{2020}b.
\newblock \bibinfo{title}{Faster approximate pattern matching: {A} unified
  approach}, in: \bibinfo{booktitle}{61st {IEEE} Annual Symposium on
  Foundations of Computer Science, {FOCS} 2020}, \bibinfo{publisher}{{IEEE}}.
  pp. \bibinfo{pages}{978--989}.
\newblock \URLprefix \url{https://doi.org/10.1109/FOCS46700.2020.00095},
  \DOIprefix\doi{10.1109/FOCS46700.2020.00095}.
\bibitem[{Crochemore et~al.(2020)Crochemore, H{\'{e}}liou, Kucherov, Mouchard,
  Pissis and Ramusat}]{DBLP:journals/iandc/CrochemoreHKMPR20}
\bibinfo{author}{Crochemore, M.}, \bibinfo{author}{H{\'{e}}liou, A.},
  \bibinfo{author}{Kucherov, G.}, \bibinfo{author}{Mouchard, L.},
  \bibinfo{author}{Pissis, S.P.}, \bibinfo{author}{Ramusat, Y.},
  \bibinfo{year}{2020}.
\newblock \bibinfo{title}{Absent words in a sliding window with applications}.
\newblock \bibinfo{journal}{Inf. Comput.} \bibinfo{volume}{270}.
\newblock \DOIprefix\doi{10.1016/j.ic.2019.104461}.
\bibitem[{Crochemore et~al.(1998)Crochemore, Mignosi and
  Restivo}]{Crochemore98automataand}
\bibinfo{author}{Crochemore, M.}, \bibinfo{author}{Mignosi, F.},
  \bibinfo{author}{Restivo, A.}, \bibinfo{year}{1998}.
\newblock \bibinfo{title}{Automata and forbidden words}.
\newblock \bibinfo{journal}{Inf. Process. Lett.} \bibinfo{volume}{67},
  \bibinfo{pages}{111--117}.
\newblock \DOIprefix\doi{10.1016/S0020-0190(98)00104-5}.
\bibitem[{Crochemore et~al.(2000)Crochemore, Mignosi, Restivo and Salemi}]{DCA}
\bibinfo{author}{Crochemore, M.}, \bibinfo{author}{Mignosi, F.},
  \bibinfo{author}{Restivo, A.}, \bibinfo{author}{Salemi, S.},
  \bibinfo{year}{2000}.
\newblock \bibinfo{title}{Data compression using antidictionaries}.
\newblock \bibinfo{journal}{Proceedings of the IEEE} \bibinfo{volume}{88},
  \bibinfo{pages}{1756--1768}.
\newblock \DOIprefix\doi{10.1109/5.892711}.
\bibitem[{Farach(1997)}]{DBLP:conf/focs/Farach97}
\bibinfo{author}{Farach, M.}, \bibinfo{year}{1997}.
\newblock \bibinfo{title}{Optimal suffix tree construction with large
  alphabets}, in: \bibinfo{booktitle}{38th Annual Symposium on Foundations of
  Computer Science, {FOCS} 1997}, \bibinfo{publisher}{{IEEE} Computer Society}.
  pp. \bibinfo{pages}{137--143}.
\newblock \URLprefix \url{https://doi.org/10.1109/SFCS.1997.646102},
  \DOIprefix\doi{10.1109/SFCS.1997.646102}.
\bibitem[{Ferenczi(1999)}]{DBLP:journals/dm/Ferenczi99}
\bibinfo{author}{Ferenczi, S.}, \bibinfo{year}{1999}.
\newblock \bibinfo{title}{Complexity of sequences and dynamical systems}.
\newblock \bibinfo{journal}{Discret. Math.} \bibinfo{volume}{206},
  \bibinfo{pages}{145--154}.
\newblock \URLprefix \url{https://doi.org/10.1016/S0012-365X(98)00400-2},
  \DOIprefix\doi{10.1016/S0012-365X(98)00400-2}.
\bibitem[{Fici and Gawrychowski(2019)}]{DBLP:conf/spire/FiciG19}
\bibinfo{author}{Fici, G.}, \bibinfo{author}{Gawrychowski, P.},
  \bibinfo{year}{2019}.
\newblock \bibinfo{title}{Minimal absent words in rooted and unrooted trees},
  in: \bibinfo{booktitle}{String Processing and Information Retrieval - 26th
  International Symposium, {SPIRE} 2019}, \bibinfo{publisher}{Springer}. pp.
  \bibinfo{pages}{152--161}.
\newblock \DOIprefix\doi{10.1007/978-3-030-32686-9\_11}.
\bibitem[{Fici et~al.(2006)Fici, Mignosi, Restivo and
  Sciortino}]{DBLP:journals/tcs/FiciMRS06}
\bibinfo{author}{Fici, G.}, \bibinfo{author}{Mignosi, F.},
  \bibinfo{author}{Restivo, A.}, \bibinfo{author}{Sciortino, M.},
  \bibinfo{year}{2006}.
\newblock \bibinfo{title}{Word assembly through minimal forbidden words}.
\newblock \bibinfo{journal}{Theor. Comput. Sci.} \bibinfo{volume}{359},
  \bibinfo{pages}{214--230}.
\newblock \DOIprefix\doi{10.1016/j.tcs.2006.03.006}.
\bibitem[{Fici et~al.(2019)Fici, Restivo and Rizzo}]{FiReRi19}
\bibinfo{author}{Fici, G.}, \bibinfo{author}{Restivo, A.},
  \bibinfo{author}{Rizzo, L.}, \bibinfo{year}{2019}.
\newblock \bibinfo{title}{Minimal forbidden factors of circular words}.
\newblock \bibinfo{journal}{Theor. Comput. Sci.} \bibinfo{volume}{792},
  \bibinfo{pages}{144--153}.
\newblock \URLprefix \url{https://doi.org/10.1016/j.tcs.2018.05.037},
  \DOIprefix\doi{10.1016/j.tcs.2018.05.037}.
\bibitem[{Fine and Wilf(1965)}]{FineWilf}
\bibinfo{author}{Fine, N.J.}, \bibinfo{author}{Wilf, H.S.},
  \bibinfo{year}{1965}.
\newblock \bibinfo{title}{Uniqueness theorems for periodic functions}.
\newblock \bibinfo{journal}{Proceedings of the American Mathematical Society}
  \bibinfo{volume}{16}, \bibinfo{pages}{109--114}.
\newblock \URLprefix \url{http://www.jstor.org/stable/2034009}.
\bibitem[{Fischer and Heun(2011)}]{DBLP:journals/siamcomp/FischerH11}
\bibinfo{author}{Fischer, J.}, \bibinfo{author}{Heun, V.},
  \bibinfo{year}{2011}.
\newblock \bibinfo{title}{Space-efficient preprocessing schemes for range
  minimum queries on static arrays}.
\newblock \bibinfo{journal}{{SIAM} J. Comput.} \bibinfo{volume}{40},
  \bibinfo{pages}{465--492}.
\newblock \URLprefix \url{https://doi.org/10.1137/090779759},
  \DOIprefix\doi{10.1137/090779759}.
\bibitem[{Fredman and Willard(1993)}]{DBLP:journals/jcss/FredmanW93}
\bibinfo{author}{Fredman, M.L.}, \bibinfo{author}{Willard, D.E.},
  \bibinfo{year}{1993}.
\newblock \bibinfo{title}{Surpassing the information theoretic bound with
  fusion trees}.
\newblock \bibinfo{journal}{J. Comput. Syst. Sci.} \bibinfo{volume}{47},
  \bibinfo{pages}{424--436}.
\newblock \DOIprefix\doi{10.1016/0022-0000(93)90040-4}.
\bibitem[{Fujishige et~al.(2016)Fujishige, Tsujimaru, Inenaga, Bannai and
  Takeda}]{DBLP:conf/mfcs/FujishigeTIBT16}
\bibinfo{author}{Fujishige, Y.}, \bibinfo{author}{Tsujimaru, Y.},
  \bibinfo{author}{Inenaga, S.}, \bibinfo{author}{Bannai, H.},
  \bibinfo{author}{Takeda, M.}, \bibinfo{year}{2016}.
\newblock \bibinfo{title}{Computing {DAWG}s and minimal absent words in linear
  time for integer alphabets}, in: \bibinfo{booktitle}{41st International
  Symposium on Mathematical Foundations of Computer Science, {MFCS} 2016},
  \bibinfo{publisher}{Schloss Dagstuhl - Leibniz-Zentrum f{\"{u}}r Informatik}.
  pp. \bibinfo{pages}{38:1--38:14}.
\newblock \DOIprefix\doi{10.4230/LIPIcs.MFCS.2016.38}.
\bibitem[{Ganguly et~al.(2018)Ganguly, Patil, Shah and
  Thankachan}]{DBLP:journals/fuin/GangulyPST18}
\bibinfo{author}{Ganguly, A.}, \bibinfo{author}{Patil, M.},
  \bibinfo{author}{Shah, R.}, \bibinfo{author}{Thankachan, S.V.},
  \bibinfo{year}{2018}.
\newblock \bibinfo{title}{A linear space data structure for range {LCP}
  queries}.
\newblock \bibinfo{journal}{Fundam. Inform.} \bibinfo{volume}{163},
  \bibinfo{pages}{245--251}.
\newblock \URLprefix \url{https://doi.org/10.3233/FI-2018-1741},
  \DOIprefix\doi{10.3233/FI-2018-1741}.
\bibitem[{Garcia et~al.(2011)Garcia, Pinho, Rodrigues, Bastos and
  Ferreira}]{minabpro}
\bibinfo{author}{Garcia, S.P.}, \bibinfo{author}{Pinho, A.J.},
  \bibinfo{author}{Rodrigues, J.M.O.S.}, \bibinfo{author}{Bastos, C.A.C.},
  \bibinfo{author}{Ferreira, P.J.S.G.}, \bibinfo{year}{2011}.
\newblock \bibinfo{title}{Minimal absent words in prokaryotic and eukaryotic
  genomes}.
\newblock \bibinfo{journal}{PLoS ONE} \bibinfo{volume}{6}.
\newblock \DOIprefix\doi{10.1371/journal.pone.0016065}.
\bibitem[{Grossi et~al.(2009)Grossi, Orlandi, Raman and
  Rao}]{grossi_et_al:LIPIcs:2009:1847}
\bibinfo{author}{Grossi, R.}, \bibinfo{author}{Orlandi, A.},
  \bibinfo{author}{Raman, R.}, \bibinfo{author}{Rao, S.S.},
  \bibinfo{year}{2009}.
\newblock \bibinfo{title}{More haste, less waste: Lowering the redundancy in
  fully indexable dictionaries}, in: \bibinfo{editor}{Albers, S.},
  \bibinfo{editor}{Marion, J.Y.} (Eds.), \bibinfo{booktitle}{26th International
  Symposium on Theoretical Aspects of Computer Science, {STACS} 2009},
  \bibinfo{publisher}{Schloss Dagstuhl--Leibniz-Zentrum fuer Informatik},
  \bibinfo{address}{Dagstuhl, Germany}. pp. \bibinfo{pages}{517--528}.
\newblock \DOIprefix\doi{10.4230/LIPIcs.STACS.2009.1847}.
\bibitem[{Harel and Tarjan(1984)}]{DBLP:journals/siamcomp/HarelT84}
\bibinfo{author}{Harel, D.}, \bibinfo{author}{Tarjan, R.E.},
  \bibinfo{year}{1984}.
\newblock \bibinfo{title}{Fast algorithms for finding nearest common
  ancestors}.
\newblock \bibinfo{journal}{{SIAM} J. Comput.} \bibinfo{volume}{13},
  \bibinfo{pages}{338--355}.
\newblock \URLprefix \url{https://doi.org/10.1137/0213024},
  \DOIprefix\doi{10.1137/0213024}.
\bibitem[{Jacobson(1989)}]{DBLP:conf/focs/Jacobson89}
\bibinfo{author}{Jacobson, G.}, \bibinfo{year}{1989}.
\newblock \bibinfo{title}{Space-efficient static trees and graphs}, in:
  \bibinfo{booktitle}{30th Annual Symposium on Foundations of Computer Science,
  {FOCS} 1989}, \bibinfo{publisher}{{IEEE} Computer Society}. pp.
  \bibinfo{pages}{549--554}.
\newblock \DOIprefix\doi{10.1109/SFCS.1989.63533}.
\bibitem[{Keller et~al.(2014)Keller, Kopelowitz, Feibish and
  Lewenstein}]{DBLP:journals/tcs/KellerKFL14}
\bibinfo{author}{Keller, O.}, \bibinfo{author}{Kopelowitz, T.},
  \bibinfo{author}{Feibish, S.L.}, \bibinfo{author}{Lewenstein, M.},
  \bibinfo{year}{2014}.
\newblock \bibinfo{title}{Generalized substring compression}.
\newblock \bibinfo{journal}{Theor. Comput. Sci.} \bibinfo{volume}{525},
  \bibinfo{pages}{42--54}.
\newblock \URLprefix \url{https://doi.org/10.1016/j.tcs.2013.10.010},
  \DOIprefix\doi{10.1016/j.tcs.2013.10.010}.
\bibitem[{Kempa and Kociumaka(2019)}]{DBLP:conf/stoc/KempaK19}
\bibinfo{author}{Kempa, D.}, \bibinfo{author}{Kociumaka, T.},
  \bibinfo{year}{2019}.
\newblock \bibinfo{title}{String synchronizing sets: sublinear-time {BWT}
  construction and optimal {LCE} data structure}, in:
  \bibinfo{booktitle}{Proceedings of the 51st Annual {ACM} {SIGACT} Symposium
  on Theory of Computing, {STOC} 2019}, \bibinfo{publisher}{{ACM}}. pp.
  \bibinfo{pages}{756--767}.
\newblock \URLprefix \url{https://doi.org/10.1145/3313276.3316368},
  \DOIprefix\doi{10.1145/3313276.3316368}.
\bibitem[{Kociumaka(2016)}]{DBLP:conf/cpm/Kociumaka16}
\bibinfo{author}{Kociumaka, T.}, \bibinfo{year}{2016}.
\newblock \bibinfo{title}{Minimal suffix and rotation of a substring in optimal
  time}, in: \bibinfo{booktitle}{27th Annual Symposium on Combinatorial Pattern
  Matching, {CPM} 2016}, pp. \bibinfo{pages}{28:1--28:12}.
\newblock \URLprefix \url{https://doi.org/10.4230/LIPIcs.CPM.2016.28},
  \DOIprefix\doi{10.4230/LIPIcs.CPM.2016.28}.
\bibitem[{Kociumaka(2018)}]{tomeksthesis}
\bibinfo{author}{Kociumaka, T.}, \bibinfo{year}{2018}.
\newblock \bibinfo{title}{Efficient Data Structures for Internal Queries in
  Texts}.
\newblock Ph.D. thesis. University of Warsaw.
\newblock \URLprefix \url{https://mimuw.edu.pl/~kociumaka/files/phd.pdf}.
\bibitem[{Kociumaka et~al.(2012)Kociumaka, Radoszewski, Rytter and
  Walen}]{DBLP:conf/spire/KociumakaRRW12}
\bibinfo{author}{Kociumaka, T.}, \bibinfo{author}{Radoszewski, J.},
  \bibinfo{author}{Rytter, W.}, \bibinfo{author}{Walen, T.},
  \bibinfo{year}{2012}.
\newblock \bibinfo{title}{Efficient data structures for the factor periodicity
  problem}, in: \bibinfo{booktitle}{String Processing and Information Retrieval
  - 19th International Symposium, {SPIRE} 2012}, pp. \bibinfo{pages}{284--294}.
\newblock \URLprefix \url{https://doi.org/10.1007/978-3-642-34109-0\_30},
  \DOIprefix\doi{10.1007/978-3-642-34109-0\_30}.
\bibitem[{Kociumaka et~al.(2015)Kociumaka, Radoszewski, Rytter and
  Walen}]{DBLP:conf/soda/KociumakaRRW15}
\bibinfo{author}{Kociumaka, T.}, \bibinfo{author}{Radoszewski, J.},
  \bibinfo{author}{Rytter, W.}, \bibinfo{author}{Walen, T.},
  \bibinfo{year}{2015}.
\newblock \bibinfo{title}{Internal pattern matching queries in a text and
  applications}, in: \bibinfo{booktitle}{Proceedings of the Twenty-Sixth Annual
  {ACM-SIAM} Symposium on Discrete Algorithms, {SODA} 2015, San Diego, CA, USA,
  January 4-6, 2015}, \bibinfo{publisher}{{SIAM}}. pp.
  \bibinfo{pages}{532--551}.
\newblock \URLprefix \url{https://doi.org/10.1137/1.9781611973730.36},
  \DOIprefix\doi{10.1137/1.9781611973730.36}.
\bibitem[{Landau and Vishkin(1988)}]{DBLP:journals/jcss/LandauV88}
\bibinfo{author}{Landau, G.M.}, \bibinfo{author}{Vishkin, U.},
  \bibinfo{year}{1988}.
\newblock \bibinfo{title}{Fast string matching with k differences}.
\newblock \bibinfo{journal}{J. Comput. Syst. Sci.} \bibinfo{volume}{37},
  \bibinfo{pages}{63--78}.
\newblock \URLprefix \url{https://doi.org/10.1016/0022-0000(88)90045-1},
  \DOIprefix\doi{10.1016/0022-0000(88)90045-1}.
\bibitem[{Matsuda et~al.(2020)Matsuda, Sadakane, Starikovskaya and
  Tateshita}]{matsuda2020compressed}
\bibinfo{author}{Matsuda, K.}, \bibinfo{author}{Sadakane, K.},
  \bibinfo{author}{Starikovskaya, T.}, \bibinfo{author}{Tateshita, M.},
  \bibinfo{year}{2020}.
\newblock \bibinfo{title}{Compressed orthogonal search on suffix arrays with
  applications to range {LCP}}, in: \bibinfo{booktitle}{31st Annual Symposium
  on Combinatorial Pattern Matching, {CPM} 2020, June 17-19, 2020, Copenhagen,
  Denmark}, pp. \bibinfo{pages}{23:1--23:13}.
\newblock \URLprefix \url{https://doi.org/10.4230/LIPIcs.CPM.2020.23},
  \DOIprefix\doi{10.4230/LIPIcs.CPM.2020.23}.
\bibitem[{Mieno et~al.(2020)Mieno, Kuhara, Akagi, Fujishige, Nakashima,
  Inenaga, Bannai and Takeda}]{DBLP:conf/sofsem/MienoKAFNIBT20}
\bibinfo{author}{Mieno, T.}, \bibinfo{author}{Kuhara, Y.},
  \bibinfo{author}{Akagi, T.}, \bibinfo{author}{Fujishige, Y.},
  \bibinfo{author}{Nakashima, Y.}, \bibinfo{author}{Inenaga, S.},
  \bibinfo{author}{Bannai, H.}, \bibinfo{author}{Takeda, M.},
  \bibinfo{year}{2020}.
\newblock \bibinfo{title}{Minimal unique substrings and minimal absent words in
  a sliding window}, in: \bibinfo{booktitle}{46th {SOFSEM}},
  \bibinfo{publisher}{Springer}. pp. \bibinfo{pages}{148--160}.
\newblock \DOIprefix\doi{10.1007/978-3-030-38919-2\_13}.
\bibitem[{Mignosi et~al.(2002)Mignosi, Restivo and
  Sciortino}]{DBLP:journals/tcs/MignosiRS02}
\bibinfo{author}{Mignosi, F.}, \bibinfo{author}{Restivo, A.},
  \bibinfo{author}{Sciortino, M.}, \bibinfo{year}{2002}.
\newblock \bibinfo{title}{Words and forbidden factors}.
\newblock \bibinfo{journal}{Theor. Comput. Sci.} \bibinfo{volume}{273},
  \bibinfo{pages}{99--117}.
\newblock \DOIprefix\doi{10.1016/S0304-3975(00)00436-9}.
\bibitem[{Munro et~al.(2020)Munro, Navarro and
  Nekrich}]{DBLP:conf/cpm/MunroNN20}
\bibinfo{author}{Munro, J.I.}, \bibinfo{author}{Navarro, G.},
  \bibinfo{author}{Nekrich, Y.}, \bibinfo{year}{2020}.
\newblock \bibinfo{title}{Text indexing and searching in sublinear time}, in:
  \bibinfo{booktitle}{31st Annual Symposium on Combinatorial Pattern Matching,
  {CPM} 2020}, \bibinfo{publisher}{Schloss Dagstuhl - Leibniz-Zentrum f{\"{u}}r
  Informatik}. pp. \bibinfo{pages}{24:1--24:15}.
\newblock \URLprefix \url{https://doi.org/10.4230/LIPIcs.CPM.2020.24},
  \DOIprefix\doi{10.4230/LIPIcs.CPM.2020.24}.
\bibitem[{Navarro(2016)}]{DBLP:books/daglib/0038982}
\bibinfo{author}{Navarro, G.}, \bibinfo{year}{2016}.
\newblock \bibinfo{title}{Compact Data Structures - {A} Practical Approach}.
\newblock \bibinfo{publisher}{Cambridge University Press}.
\bibitem[{Ota and Morita(2010)}]{DBLP:conf/isita/OtaM10}
\bibinfo{author}{Ota, T.}, \bibinfo{author}{Morita, H.}, \bibinfo{year}{2010}.
\newblock \bibinfo{title}{On the adaptive antidictionary code using minimal
  forbidden words with constant lengths}, in: \bibinfo{booktitle}{Proceedings
  of the International Symposium on Information Theory and its Applications,
  {ISITA} 2010}, \bibinfo{publisher}{{IEEE}}. pp. \bibinfo{pages}{72--77}.
\newblock \DOIprefix\doi{10.1109/ISITA.2010.5649621}.
\bibitem[{P{\v a}tra{\c s}cu and Thorup(2014)}]{DBLP:conf/focs/PatrascuT14}
\bibinfo{author}{P{\v a}tra{\c s}cu, M.}, \bibinfo{author}{Thorup, M.},
  \bibinfo{year}{2014}.
\newblock \bibinfo{title}{Dynamic integer sets with optimal rank, select, and
  predecessor search}, in: \bibinfo{booktitle}{55th {IEEE} Annual Symposium on
  Foundations of Computer Science, {FOCS} 2014}, \bibinfo{publisher}{{IEEE}
  Computer Society}. pp. \bibinfo{pages}{166--175}.
\newblock \URLprefix \url{https://doi.org/10.1109/FOCS.2014.26},
  \DOIprefix\doi{10.1109/FOCS.2014.26}.
\bibitem[{Pratas and Silva(2020)}]{10.1093/bioinformatics/btaa686}
\bibinfo{author}{Pratas, D.}, \bibinfo{author}{Silva, J.M.},
  \bibinfo{year}{2020}.
\newblock \bibinfo{title}{{Persistent minimal sequences of SARS-CoV-2}}.
\newblock \bibinfo{journal}{Bioinformatics}
  \DOIprefix\doi{10.1093/bioinformatics/btaa686}. \bibinfo{note}{btaa686}.
\bibitem[{Raskhodnikova et~al.(2013)Raskhodnikova, Ron, Rubinfeld and
  Smith}]{DBLP:journals/algorithmica/RaskhodnikovaRRS13}
\bibinfo{author}{Raskhodnikova, S.}, \bibinfo{author}{Ron, D.},
  \bibinfo{author}{Rubinfeld, R.}, \bibinfo{author}{Smith, A.D.},
  \bibinfo{year}{2013}.
\newblock \bibinfo{title}{Sublinear algorithms for approximating string
  compressibility}.
\newblock \bibinfo{journal}{Algorithmica} \bibinfo{volume}{65},
  \bibinfo{pages}{685--709}.
\newblock \URLprefix \url{https://doi.org/10.1007/s00453-012-9618-6},
  \DOIprefix\doi{10.1007/s00453-012-9618-6}.
\bibitem[{Rubinchik and Shur(2017)}]{DBLP:conf/spire/RubinchikS17}
\bibinfo{author}{Rubinchik, M.}, \bibinfo{author}{Shur, A.M.},
  \bibinfo{year}{2017}.
\newblock \bibinfo{title}{Counting palindromes in substrings}, in:
  \bibinfo{booktitle}{String Processing and Information Retrieval - 24th
  International Symposium, {SPIRE} 2017}, \bibinfo{publisher}{Springer}. pp.
  \bibinfo{pages}{290--303}.
\newblock \URLprefix \url{https://doi.org/10.1007/978-3-319-67428-5\_25},
  \DOIprefix\doi{10.1007/978-3-319-67428-5\_25}.
\bibitem[{Silva et~al.(2015)Silva, Pratas, Castro, Pinho and
  Ferreira}]{Silva02042015}
\bibinfo{author}{Silva, R.M.}, \bibinfo{author}{Pratas, D.},
  \bibinfo{author}{Castro, L.}, \bibinfo{author}{Pinho, A.J.},
  \bibinfo{author}{Ferreira, P.J.S.G.}, \bibinfo{year}{2015}.
\newblock \bibinfo{title}{Three minimal sequences found in {E}bola virus
  genomes and absent from human {DNA}}.
\newblock \bibinfo{journal}{Bioinform.} \bibinfo{volume}{31},
  \bibinfo{pages}{2421--2425}.
\newblock \URLprefix \url{https://doi.org/10.1093/bioinformatics/btv189},
  \DOIprefix\doi{10.1093/bioinformatics/btv189}.
\bibitem[{Tanimura et~al.(2017)Tanimura, Nishimoto, Bannai, Inenaga and
  Takeda}]{DBLP:conf/mfcs/TanimuraNBIT17}
\bibinfo{author}{Tanimura, Y.}, \bibinfo{author}{Nishimoto, T.},
  \bibinfo{author}{Bannai, H.}, \bibinfo{author}{Inenaga, S.},
  \bibinfo{author}{Takeda, M.}, \bibinfo{year}{2017}.
\newblock \bibinfo{title}{Small-space {LCE} data structure with constant-time
  queries}, in: \bibinfo{booktitle}{42nd International Symposium on
  Mathematical Foundations of Computer Science, {MFCS} 2017},
  \bibinfo{publisher}{Schloss Dagstuhl - Leibniz-Zentrum f{\"{u}}r Informatik}.
  pp. \bibinfo{pages}{10:1--10:15}.
\newblock \URLprefix \url{https://doi.org/10.4230/LIPIcs.MFCS.2017.10},
  \DOIprefix\doi{10.4230/LIPIcs.MFCS.2017.10}.
\bibitem[{Tiskin(2008)}]{DBLP:journals/mics/Tiskin08}
\bibinfo{author}{Tiskin, A.}, \bibinfo{year}{2008}.
\newblock \bibinfo{title}{Semi-local string comparison: Algorithmic techniques
  and applications}.
\newblock \bibinfo{journal}{Math. Comput. Sci.} \bibinfo{volume}{1},
  \bibinfo{pages}{571--603}.
\newblock \URLprefix \url{https://doi.org/10.1007/s11786-007-0033-3},
  \DOIprefix\doi{10.1007/s11786-007-0033-3}.
\bibitem[{Yao(1982)}]{Yao1982}
\bibinfo{author}{Yao, A.C.}, \bibinfo{year}{1982}.
\newblock \bibinfo{title}{Space-time tradeoff for answering range queries
  (extended abstract)}, in: \bibinfo{booktitle}{Proceedings of the Fourteenth
  Annual ACM Symposium on Theory of Computing}, \bibinfo{publisher}{ACM}. pp.
  \bibinfo{pages}{128--136}.
\newblock \DOIprefix\doi{10.1145/800070.802185}.

\end{thebibliography}

\end{document}